\documentclass[10pt,a4paper,amsmath,amssymb,pdflatex,showpacs,pra,onecolumn,notitlepage,nopdfoutputerror]{quantumarticle}

\usepackage{dcolumn}
\usepackage{mathtools}
\usepackage{enumerate}
\usepackage[hiresbb]{graphicx}
\usepackage{dcolumn}
\usepackage{bm}
\usepackage{physics}
\usepackage{comment}
\usepackage{booktabs}
\usepackage{color}
\usepackage{hyperref}
\usepackage{thm-restate} 
\newcounter{one}
\allowdisplaybreaks

\def\ketbra#1#2{\ket{#1}\!\bra{#2}}

\newcommand{\expect}[1]{\langle #1 \rangle}

\newtheorem{theorem}{Theorem}

\newtheorem{definition}{Definition}

\newtheorem{proposition}{Proposition}
\newtheorem{remark}{Remark}

\def\QED{\mbox{\rule[0pt]{1.5ex}{1.5ex}}}

\def\endproof{\hspace*{\fill}~\QED\par\endtrivlist\unskip}

\newcommand{\eq}[1]{\begin{align}
  #1
\end{align}}

\newcommand\calD{{\cal D}}
\newcommand\calE{{\cal E}}
\newcommand\calF{{\cal F}}

\newcommand\calH{{\cal H}}

\newcommand\calL{{\cal L}}

\newcommand\calU{{\cal U}}

\newcommand{\beq}{\begin{equation}}
\newcommand{\eeq}{\end{equation}}

\newcommand{\revision}[1]{{\color{black}{#1}}}

\begin{document}
\title{Improvement of Speed Limits: Quantum Effect on the Speed in Open Quantum Systems}
\author{Kotaro Sekiguchi}
\email{sekiguchi.kotaro.258@s.kyushu-u.ac.jp}
\affiliation{Graduate School of Information Science and Electrical Engineering, Kyushu University, Japan}
\author{Satoshi Nakajima}
\affiliation{}
\author{Ken Funo}
\affiliation{Department of Applied Physics, The University of Tokyo, 7-3-1 Hongo, Bunkyo-ku, Tokyo 113-8656, Japan}
\author{Hiroyasu Tajima}
\email{hiroyasu.tajima@inf.kyushu-u.ac.jp}
\affiliation{Department of Information Science and Technology, Kyushu University, Nishi-ku, Fukuoka, 819-0395, Japan}
\affiliation{JST, FOREST, 4-1-8 Honcho, Kawaguchi, Saitama, 332-0012, Japan}

\begin{abstract}
In the context of quantum speed limits, it has been shown that the minimum time required to cause a desired state conversion via the open quantum dynamics can be estimated using the entropy production. 
However, the established entropy-based bounds tend to be loose, making it difficult to accurately estimate the minimum time for evolution. 
In this research, we have combined the knowledge of the entropy-based speed limits with that of the resource theory of asymmetry (RTA) and provided much stricter inequalities.
Our results show that the limitation on the change rate of states and expectation values can be divided into two parts: quantum coherence for energy (i.e., asymmetry) contributed by the system and the heat bath and the classical entropy-increasing effect from the bath.
As a result, our inequalities demonstrate that the difference in the speed of evolution between classical and quantum open systems, i.e., the quantum enhancement in speed, is determined by the quantum Fisher information, which measures quantum fluctuations of energy and serves as a standard resource measure in the resource theory of asymmetry. 
We further show that a similar relation holds for the rate of change of expectation values of physical quantities.
\end{abstract}

\maketitle	

\section{Introduction}
The Quantum Speed Limits (QSLs) are inequalities that evaluate the minimum time required to transition from one quantum state to another. 
\revision{They were initially formulated for isolated quantum systems \cite{mandelstam1991uncertainty, margolus1998maximum,Braustein1994-geometry,UHLMANN1992,Anandan1990-geometry,Fleming1973-unitarybound, PhysRevLett.91.180403} as rigorous formulations of the time-energy uncertainty relations, and afterwards they have been extended to open quantum systems \cite{openQSL1,openQSL2,openQSL3,FSS,VS, Nakajima2022,Deffner2017,Pintos2022,MONDAL2016689,Abhay2024,Lokutsievskiy_2024}.}
They provide universal constraints on the rates of state changes across various dynamics, regardless of the specific characteristics of the system. This universality of the QSLs makes them invaluable for advancing fields such as quantum computing \cite{Lloyd2000}, quantum metrology \cite{PhysRevLett.112.120405}, quantum optical control \cite{opt1,opt2,opt3}, and \revision{quantum thermodynamics \cite{PhysRevLett.118.100602,PhysRevLett.125.120604,Nicholson2020,PhysRevLett.105.170402,PhysRevE.99.042105}.}

The universality of quantum speed limits reflects their origin in the fundamental geometric properties of quantum state space \cite{Pires2016-geometric,Braustein1994-geometry,Anandan1990-geometry,Deffner2017,MONDAL2016689}. When an arbitrary distance $D$ between quantum states is introduced, the distance $D(\rho, \rho(\tau))$ between the state $\rho$ at time $0$ and the state $\rho(\tau)$ at time $\tau$ is always bounded by the integral of the metric $g$ induced by $D$, as $D(\rho, \rho(\tau)) \leq \int_0^\tau \sqrt{g_{tt}} \, dt$.
Many speed limits for open quantum systems \cite{openQSL1,openQSL2,openQSL3} have essentially been derived based on this idea.
However, the metric of state space itself does not necessarily have \revision{a} connection to physically meaningful quantities, and the resulting inequalities have not bounded $D(\rho, \rho(\tau))$ in terms of easily interpretable quantities. Consequently, the entropy-based quantum speed limits have recently been given \cite{FSS,VS}, inspired by speed limits for classical open systems \cite{SFS,Ohzeki}.
These entropy-based quantum speed limits constrain the rate of state change by two quantities with clear physical meanings: the thermodynamic entropy production rate and energy variance. 
\revision{The first quantity, the entropy production rate \cite{breuer2002theory,Santos2019}}, represents the irreversible changes due to interactions with a heat bath. The term containing the entropy production rate in the QSLs has a classical counterpart in the speed limits of classical open systems \cite{SFS,Ohzeki}. The second quantity, energy variance, does not appear in the classical speed limit, and therefore the energy variance has been considered to describe the quantum effect in the QSL for open quantum systems. 
By using these quantities, the entropy-based QSLs have succeeded in providing intuitive constraints on the time required for state transformation.

Despite these advantages, there is a concern about the existing entropy-based QSLs for open quantum systems; they are not sufficiently tight.
Since the primary goal of QSLs is to constrain the minimum state-transition time, the tightness of these inequalities is a critical issue. 
To tackle this problem, we pay attention to the energy variance, which quantifies the energy fluctuation originating from both the quantum superpositions and classical probability mixtures. 
Since the energy fluctuation does not constrain the classical speed limit, we expect that only the quantum part of the energy fluctuation is relevant to the QSLs. 
This inference is supported by numerical calculations (see Fig.~\ref{fig:kaizen}). The figure shows that neither of the existing inequalities provides a good bound on the actual rate of state changes. 

In this article, based on the above considerations, we improve the speed limit for state transitions in open quantum systems. 
Our speed limit uses the quantum fluctuations of energy and the entropy production rate. 
As an indicator of quantum fluctuations, we employ the quantum Fisher information (also known as the skew information), a standard resource measure in the resource theory of asymmetry~\cite{Gour2008resource,skew_resource,Marvian_thesis,Marvian_distillation,Takagi_skew,Marvian2022operational,Kudo,YT,YT2,ST}, a resource theory generally treating the quantum fluctuations of conserved quantities including energy as \revision{resources}. 
\revision{In particular, for isolated quantum systems, the Wigner–Yanase skew information has been used to estimate the speed of quantum-state evolution \cite{PhysRevLett.91.180403}.}
Recent developments in uncertainty relations \cite{Luo,Hansen,gibilisco2007uncertainty, gibilisco2011refinement} and the quantum Fisher information \cite{min_V_Petz,min_V_Yu} have revealed that the quantum Fisher information corresponds to the ``quantum fluctuation part" in the energy variance.
Therefore, our bound can be understood as a tighter speed limit for open quantum systems in which the classical fluctuations are removed from the energy variance terms that appeared in the previous bounds. 
We also improved the ``classical dissipation part" of the existing speed limit, i.e. the term of the entropy production, by using an improved version of the mobility.
The numerical calculations confirm that our bound is significantly tighter than the previous bounds.

Furthermore, we provide similar evaluations for currents, with the speed limits of state changes being a special case. 
The obtained bounds relate the currents of physical quantities to the quantum coherence of the energy and the entropy production.
Our inequalities reflect the decomposition of the speeds of changes of states and expectation values of physical quantities in open quantum systems into components attributable to isolated system dynamics and irreversible environmental changes, demonstrating that these speeds are governed solely by quantum fluctuations of energy and the increase in entropy due to the heat bath.

\section{Preliminaries}\label{sec:preliminary}
In this section, we introduce the prerequisite knowledge necessary to describe the results of this paper: the quantum master equation (i.e., the Gorini–Kossakowski–Sudarshan–Lindblad: GKSL equation), thermodynamic quantities in the stochastic thermodynamics, the methods to quantify the quantum superposition between energy eigenstates in the resource theory of asymmetry, and the previous entropy-based speed limits for open quantum systems.

\subsection{GKSL equation}
We consider an open quantum system obeying the following \revision{GKSL} equation \cite{breuer2002theory}:
\begin{align}
  \dot{\rho}_t = \calL_t[\rho(t)]= -\frac{i}{\hbar}[H(t),\rho(t)] + \calD[\rho(t)] \label{eq:Lindblad}.
\end{align}
Here, $H(t)$ is the system Hamiltonian, \revision{$\calL_t$ is the Liouvillian that generates the time evolution of the density operator, and $\calD$ is the superoperator (a linear operator acting on operators) that describes dissipation arising from the interaction with the external system.
The dissipative term $\calD[\rho(t)]$ is expressed as follows:}
\begin{align}
  \calD[\rho(t)] = \sum_{\omega_t,\alpha}\gamma_{\omega_t,\alpha} \qty( L_{\omega_t,\alpha} \rho(t) L_{\omega_t,\alpha}^\dag 
  - \frac{1}{2} \qty{ L_{\omega_t,\alpha}^\dag L_{\omega_t,\alpha}, \rho(t) } ) .
  \label{eq:Lindblad_dissipation}
\end{align}
\revision{%
In the above equation, $L_{\omega_t,\alpha}$ are called the Lindblad (jump) operators, and $\gamma_{\omega_t, \alpha} \geq 0$ are the transition rate coefficients.
The label $\omega_t$ denotes an instantaneous energy gap of $H(t)$, while $\alpha$ distinguishes different jump channels.
The jump operators and rate coefficients may depend on time, with this dependence left implicit in the notation.
}

\revision{
We assume that the jump operators satisfy
\begin{align}
    [L_{\omega_t,\alpha}, H(t)] = \omega_t L_{\omega_t,\alpha},
    \qquad
    L_{-\omega_t,\alpha} = L_{\omega_t,\alpha}^\dagger.
\end{align}
The first relation means that $L_{\omega_t, \alpha}$ induces transitions between energy eigenspaces whose energies differ by $\omega_t$, lowering the system energy by $\omega_t$ when $\omega_t > 0$.
The second relation identifies $L_{-\omega_t,\alpha}$ as the jump operator for the reverse transition.
We further assume that the rate coefficients satisfy the detailed balance condition
\begin{align}
    \gamma_{-\omega_t,\alpha} = \gamma_{\omega_t,\alpha}e^{-\beta\omega_t}, \label{eq:detail_balance}
\end{align}
where $\beta$ is the inverse temperature of the heat bath.
All three conditions are imposed at each time $t$.
Under the above assumptions on the jump operators, the detailed balance condition Eq.~\eqref{eq:detail_balance} is sufficient to ensure that the Gibbs state at inverse temperature $\beta$ is an instantaneous stationary state of the system.
}

For the \revision{later} convenience, we remark that we can decompose the dissipative term $\calD[\rho(t)]$ into a ``unitary part" and the ``classical dissipative part" \cite{FSS}.
Using the \revision{spectral decomposition} of \revision{$\rho(t)=\sum_n p_n(t) \Pi_n(t)$}, we can write $\calD\qty[\rho(t)] = \calD_d\qty[\rho(t)] + \calD_{nd}\qty[\rho(t)]$. 
Here, $\calD_d\qty[\rho(t)]$ and $\calD_{nd}\qty[\rho(t)]$ are defined as

\revision{
\begin{align}
  \calD_d\qty[\rho(t)] &:= \sum_n \Pi_n(t) \calD[\rho(t)] \Pi_n(t), \\
  \calD_{nd}\qty[\rho(t)] &:= \sum_{m\neq n} \Pi_m(t) \calD[\rho(t)] \Pi_n(t),
\end{align}}%
representing the diagonal and off-diagonal components of the dissipative term in the diagonal basis of $\rho(t)$.
The important point is that $\calD_{nd}[\rho(t)]$ can be understood as the unitary part of $\calD[\rho(t)]$, because we can construct a Hermitian operator $H_{\calD}$ which works as an effective Hamiltonian of $\calD_{nd}[\rho(t)]$ as follows:
\begin{align}
  \calD_{nd}[\rho] = - \frac{i}{\hbar} \qty[H_\calD(t), \rho(t)].
\end{align}
The concrete definition of $H_{\calD}$ is as follows:
\revision{
\begin{align}
  H_\calD(t) \revision{:=} \sum_{m\neq n} \frac{i\hbar }{p_n - p_m}\Pi_m\calD\qty[\rho(t)]\Pi_n. \label{eq:H_D}
\end{align}}%
Therefore, the \revision{GKSL} equation is expressed in the following form:
\begin{align}
  \dot{\rho} = -\frac{i}{\hbar}[H(t)+H_{\calD}(t),\rho(t)] + \calD_d\qty[\rho(t)] .\label{eq:alt_decom}
\end{align}

\subsection{Thermodynamic quantities}
Next, we define the thermodynamic quantities in open quantum systems \cite{breuer2002theory} used in this paper.
We define the von Neumann entropy flux and the heat flux, and based on these two quantities, we define the entropy production rate.
The von Neumann entropy flux is defined as follows:
\begin{align}
  \dot{S} &:= -\Tr \qty[\dot{\rho}_t \ln \rho_t] .\label{eq:v.N.ent_flux}
\end{align}
The von Neumann entropy flux represents the entropy change of the primary system.
The heat flux is defined as follows:
\begin{align}
  \dot{Q} &:= \Tr \qty[\dot{\rho}_t H(t)]. \label{eq:heat_flux}
\end{align}
The heat flux is used to represent the energy change rate of the system due to the change of the density matrix.
When we assume the total energy conservation of the system and the heat bath, the minus of the heat flux $-\dot{Q}$ is equal to the energy change rate of the heat bath.
Using these quantities, the entropy production rate is defined as follows:
\begin{align}
  \dot{\sigma} := \dot{S} - \beta\dot{Q} .\label{eq:entropy_production}
\end{align}
Here, $\beta$ represents the inverse temperature of the bath.

\revision{The entropy production rate has two important interpretations.}
First, this quantity can be interpreted as the rate of \revision{change of the thermodynamic entropy of the composite system consisting of the system and the heat bath.}
\revision{In the GKSL equation description, the relaxation \revision{time of the heat bath is assumed to be much shorter than the typical timescale of the system dynamics.}}
In other words, the heat bath \revision{exchanges heat while remaining close to thermal equilibrium.} 
Consequently, we can \revision{identify the rate of entropy change of the heat bath as } $-\beta \dot{Q}$. 
Therefore, $\dot{\sigma}$ is the sum of the \revision{rates of entropy change of the system and the heat bath, and hence equals the rate of change of the thermodynamic entropy of the total system.} 

\revision{
Second, this quantity represents the rate at which the relative entropy between the state $\rho_t$ and the instantaneous Gibbs state decreases under the evolution generated by the Liouvillian held fixed at time $t$.
}%
\revision{%
More precisely, let $\calL_t$ denote the Liouvillian at time $t$ in the GKSL equation Eq.~\eqref{eq:Lindblad}.
The entropy production rate $\dot{\sigma}$ then satisfies the following relation \cite{spohn-Lebowitz1978}:
}
\revision{
\begin{align}
\dot{\sigma}=\lim_{dt\rightarrow 0} \frac{D(\rho_t\|\rho_{\beta|H(t)})-D(e^{\calL_t dt}\rho_t \| e^{\calL_t dt}\rho_{\beta|H(t)})}{dt}.
\end{align}
}%
\revision{%
Here, $\rho_{\beta|H}$ is the Gibbs state for the Hamiltonian $H(t)$ at inverse temperature $\beta$, and the limit is taken over positive time intervals $dt>0$.
}%
\revision{%
The map $e^{\calL_t dt}$ is a completely positive and trace-preserving (CPTP) map describing the evolution generated by $\calL_t$ held fixed at time $t$.
}%
\revision{%
Therefore, the monotonicity of relative entropy under CPTP maps implies the non-negativity of the entropy production rate, $\dot{\sigma} \geq 0$, and thus the second law of thermodynamics is satisfied in this setting.
}
The entropy production rate can be represented using the following state transition \revision{rate} \cite{FSS}:
\revision{
\begin{align}
  W^{\omega,\alpha}_{ab} := \gamma_{\omega,\alpha}\Tr\qty[\Pi_a L_{\omega,\alpha}\Pi_b L^\dag_{\omega,\alpha}] .\label{eq:rate_matrix}
\end{align}
}%
Here, the quantum state at time $t$ can be \revision{spectrally decomposed as $\rho_t=\sum_n p_n(t) \Pi_n(t)$}, and we denote them as \revision{$p_n,\Pi_n$}. %
\revision{%
We now express the entropy production rate in terms of the transition rates.
Using Eqs.~\eqref{eq:Lindblad} and \eqref{eq:rate_matrix}, we first rewrite the rate of change of the von Neumann entropy and the heat flux as follows \cite{FSS} (see Appendix~\ref{Ap:projection_op} for details of the derivation):
}
\begin{align}
  -\beta\dot{Q} &= \sideset{}{^\prime}\sum_{\omega,\alpha,n,m} W^{\omega,\alpha}_{mn} p_n \ln \frac{W^{\omega,\alpha}_{mn}}{W^{-\omega,\alpha}_{nm}},\\
  \dot{S} &= - \sum_{\omega,\alpha,n,m} W^{\omega,\alpha}_{mn} p_n \ln p_m + \sum_{\omega,\alpha,n,m} W^{\omega,\alpha}_{mn} p_n \ln p_n.
\end{align}
Here, $\beta\omega=\ln \frac{W^{\omega,\alpha}_{mn}}{W^{-\omega,\alpha}_{nm}}$ is derived from the detailed balance condition. Furthermore, $\sideset{}{^\prime}\sum$ means that we sum over all $\omega,\alpha,n,m$ except for $(m=n)\land (\omega=0)$. Thus, the entropy production rate can be represented as follows:
\begin{align}
  \dot{\sigma} = \sideset{}{^\prime}\sum_{\omega,\alpha,n,m} W^{\omega,\alpha}_{mn} p_n \ln \frac{W^{\omega,\alpha}_{mn}p_n}{W^{-\omega,\alpha}_{nm}p_m} \ .  \label{eq:ent_production2}
\end{align}
The following relation holds:
\begin{align}
\sideset{}{^\prime}\sum_{\omega,\alpha,n,m} W^{\omega,\alpha}_{mn}p_n = \sideset{}{^\prime}\sum_{\omega,\alpha,n,m} W^{-\omega,\alpha}_{nm}p_m. \label{eq:unnormalized_distribution}
\end{align}
\revision{In other words, $\dot{\sigma}$ is the relative entropy of the unnormalized distribution $\{W^{\omega,\alpha}_{mn} p_n\}$ with respect to $\{W^{-\omega,\alpha}_{nm}p_m\}$, both of which have equal total weight by Eq.~\eqref{eq:unnormalized_distribution}.
Since the relative entropy between two such distributions is non-negative \cite{Shiraishi2019}, we obtain $\dot{\sigma}\geq 0$ ---an alternative proof of the second law of thermodynamics.
}

\revision{
\subsection{Quantum fluctuations and general inner products}
}
The main purpose of this paper is to investigate the influence of quantum fluctuations on the speed limit of the currents of physical quantities in open quantum systems and to improve the speed limit of open quantum systems.

There is a criterion for quantum fluctuations given by Luo \cite{Luo}. 

\begin{definition}
  A quantum fluctuation is defined as a function $Q_\rho(H)$ that satisfies the following conditions for a quantum state $\rho$ and a physical quantity $H$:
  \begin{enumerate}
    \item $0 \leq Q_\rho(H) \leq \revision{\mathrm{Var}_\rho(H)}$.
    \item $\rho$ is a pure state $\Rightarrow Q_\rho(H) = \revision{\mathrm{Var}_\rho(H)}$.
    \item $\qty[\rho, H] = 0 \Rightarrow Q_\rho(H) = 0$.
    \item $Q_{\sum_i p_i \rho_i}(H) \leq \sum_i p_i Q_{\rho_i}(H)$.
  \end{enumerate}
  \revision{Here, $\mathrm{Var}_\rho(H)$ denotes the ordinary variance of the observable $H$ in the state $\rho$.
  Its generalization, the $f$-variance $V_\rho^f(H)$, is defined in Eq.~\eqref{eq:variance_f} using the $f$-inner product introduced below.}
\end{definition}
The first condition says that $Q_\rho(H)$ is a part of the variance.
The second condition requires that when the state is pure, namely when all of the fluctuation of $H$ is caused by the quantum superposition, $Q_\rho(H)$ is equal to the variance.
The third condition means that when the state $\rho$ commutes with $H$, namely when all of the fluctuation of $H$ is caused by the classical mixture, $Q_\rho(H)$ is equal to 0.
The last condition is the monotonicity of $Q_\rho(H)$ under the classical mixture.
Due to these four conditions, we can consider $Q_\rho(H)$ as the ``quantum part" of the variance.
Although there are infinitely many quantities that satisfy these four conditions (see below), these four conditions still work as ``necessary conditions" for the indicators of quantum fluctuation.

When a quantity satisfies Luo's criterion, the variance can be decomposed into ``fluctuations due to quantum superposition" and ``fluctuations due to classical mixture":
\eq{
\revision{\mathrm{Var}_\rho(H)}=Q_\rho(H)+C_\rho(H).\label{variance-decom}
}
This $C_\rho(H):=\mathrm{Var}_\rho(H)-Q_\rho(H)$ is always non-negative due to Luo's criterion 1, always 0 for pure state by Luo's criterion 2, and equal to the variance when $\rho$ and $H$ commute (when all variance is due to classical \revision{mixture}) due to Luo's criterion 3. Because of these properties, $C_\rho(H)$ can be interpreted as the ``fluctuations due to classical probability".
This fact is a good explanation of why our main results improve previous studies.

\revision{
Next, we present the quantum Fisher information as important examples of $Q_\rho(H)$.
Although there exist infinitely many versions of quantum Fisher information depending on the choice of inner product, it is known that a large class of quantum Fisher informations can be defined in a unified way using a standard operator monotone function $f$ \cite{PETZ_2011}.
}

\revision{
We define a standard operator monotone function as follows.
\begin{definition}
  For a function $f: (0, \infty) \rightarrow (0, \infty)$,
  \begin{enumerate}
    \item $f(1) = 1$,
    \item $f(x) = xf(x^{-1})$,
    \item $0 \leq A \leq B \Rightarrow f(A) \leq f(B)$
  \end{enumerate}
  are satisfied, the function $f$ is called a standard operator monotone function.
\end{definition}
This function $f$ is used to define the inner product in general.
}
\revision{
\begin{definition}
  For a standard operator monotone function $f: (0, \infty) \rightarrow (0, \infty)$, a state $\rho$, and \revision{Hermitian} operators $A$ and $B$, the inner product is defined as follows:
  \begin{align}
    \langle A, B \rangle_{\rho}^f \coloneqq \Tr \qty[ A \mathbb{J}^f_{\rho} (B) ]. \label{eq:inner_product}
  \end{align}
  Here, $\mathbb{J}^f_{\rho}$ is a superoperator defined as follows:
  \begin{equation}
    \mathbb{J}^f_{\rho} = f(\mathbb{L}_\rho \mathbb{R}^{-1}_\rho)\mathbb{R}_\rho, \quad \mathbb{L}_\rho (X) = \rho X, \quad \mathbb{R}_\rho (X) = X \rho.
  \end{equation}
  The inner product defined here (referred to as the $f$-inner product) has the following properties:
  \begin{itemize}
    \item $\langle A, B \rangle_{\rho}^f = (\langle B, A \rangle_{\rho}^f)^*$.
    \item $\langle A, \sum_i b_i B_i \rangle_{\rho}^f = \sum_i b_i \langle A, B_i\rangle_{\rho}^f $.
    \item $\langle A, A \rangle_{\rho}^f \geq 0$.
  \end{itemize}  
\end{definition}
   Using this inner product, the \revision{$f$-}Fisher information and the \revision{$f$-}variance can be defined in general:
   }
\revision{
\begin{definition}
Fisher information is defined as follows,
  \eq{
  J^f_{\rho_t} = \langle L,L \rangle_{\rho_t}^f .
  }
  $L$ is a logarithmic derivative and it is defined by $\frac{\partial \rho_t}{\partial t}= \mathbb{J}^f_{\rho_t}(L)$.
  Additionally, the \revision{$f$-}variance is defined as follows,
  \begin{align}
    V_{\rho}^f(A) \coloneqq \langle A_0, A_0 \rangle_{\rho}^f. \label{eq:variance_f}
  \end{align}
  Here, we define $A_0$ as
  \begin{equation}
    A_0 \coloneqq A - \langle A \rangle_{\rho}.
  \end{equation}
\end{definition}
It is known that the following Cauchy-Schwarz inequality holds for the inner product defined in this way,
\begin{align}
\abs{\langle A, B\rangle_\rho^f}^2 \leq \langle A, A\rangle_\rho^f \langle B, B\rangle_\rho^f.
\end{align}
This inequality is used in deriving the main result of this study.
}

\revision{
Using an operator monotone function $f$, the quantum Fisher information for a family of quantum states $\qty{e^{-iHt} \rho e^{iHt}}$ is defined as follows (see Appendix \ref{Ap:Fisher_general} for details):
}
\revision{
\begin{theorem}
  If the spectral decomposition $\rho =\sum_i p_i \ketbra{i}{i}$, then $\calF_\rho^f(H)$ for quantum state family $\qty{e^{-iHt} \rho e^{iHt}}$ can be calculated as follows:
  \begin{align}
    \calF_\rho^f(H) = \sum_{i,j} \frac{(p_i - p_j)^2}{p_j f(\frac{p_i}{p_j})} \abs{ \bra{i} H \ket{j} }^2 . \label{eq:fisher-state-family}
  \end{align}
\end{theorem}
}

\revision{The relationship between the variance $V^f_\rho(H)$ and the normalized Fisher information $\frac{f(0)}{2}\calF^f_\rho(H)$ satisfies Luo's criteria.}
\revision{From this fact, we see that the quantum Fisher information has the meaning of a quantum fluctuation of an observable. 
This interpretation is particularly clear in the case of the SLD (Symmetric Logarithmic Derivative) Fisher information, i.e., when $f(x)=(1+x)/2$, for which the following theorem holds:
\begin{restatable}{theorem}{SLDdecomposition}
  \label{theo:SLD_Fisher2}
  For a quantum state $\rho$, the SLD Fisher information is equal to four times the minimum average variance of the pure states in the ensemble $\qty{p_i,\ket{\phi_i}}$ of the quantum state $\rho$:
  \begin{align}
    \calF_\rho(H) = 4 \min_{p_i,\ket{\phi_i}} \sum_i p_i V_{\ket{\phi_i}}(H).
  \end{align}
  Here, the minimum is taken over all ensembles, including nonorthogonal ones, satisfying \\$\rho=\sum_ip_i\ket{\phi_i}\bra{\phi_i}$.
\end{restatable}
In other words, apart from the constant factor of $4$, the SLD Fisher information is equal to the minimum average variance among all pure-state decompositions of $\rho$.
}

Furthermore, the SLD Fisher information is a standard resource measure in the resource theory of asymmetry (RTA) for the $U(1)$ case \cite{Marvian2022operational}. We will discuss this in detail in Appendix \ref{Ap:RTA}.
\revision{Additionally, we denote the variance and Fisher information defined in particular by the SLD inner product as $V_\rho$ and $\calF_\rho$, respectively.}

\subsection{The entropy-based speed limits in open quantum systems}\label{sec:open_speed_limit}
One of the goals of this paper is to improve the existing entropy-based quantum speed limits.
Therefore, we here introduce the conventional entropy-based quantum speed limits.
The first one is given as follows \cite{FSS}:
\revision{
\begin{align}
  \norm{\dot{\rho}}_{1} \leq \frac{2}{\hbar} \sqrt{V_\rho(H)} + \frac{2}{\hbar} \sqrt{V_\rho(H_\calD)}  + \sqrt{2\dot{\sigma}A}.\label{eq:FSS}
\end{align}
}%
\revision{Here, $ \norm{X}_{1}:= \Tr \qty[\sqrt{X^\dagger X}]$.}
The quantity $A$ is the activity that is defined as follows:
\begin{align}
  A :=& \frac{1}{2}\sideset{}{^\prime}\sum_{\omega,\alpha,n, m} \qty(W^{\omega,\alpha}_{m,n}p_n + W^{-\omega,\alpha}_{n,m}p_m).\label{eq:activity}
\end{align}

The activity represents how frequently the state transitions.
We call the bound Eq.~\eqref{eq:FSS} the Funo-Shiraishi-Saito (FSS) bound hereafter.

The second one is given as an improvement of the FSS bound. This bound is given by Vu and Saito \cite{VS}:
\revision{
\begin{align}
  \norm{\dot{\rho}}_{1} \leq \frac{2}{\hbar} \sqrt{V_\rho(H)} + \frac{2}{\hbar} \sqrt{V_\rho(H_\calD)}  + \sqrt{2\dot{\sigma}M}.\label{eq:VS}
\end{align}
}%
We call this the Vu-Saito (VS) bound hereafter.
Here, $M$ is a quantity called mobility:
\eq{
M := \revision{\sum_{\omega,\alpha,n\neq m} }\revision{g(W^{\omega,\alpha}_{m,n}p_n, W^{-\omega,\alpha}_{n,m}p_m)}. \label{eq:mobility}
}
Here \revision{$g(a,b):=(a-b)/\log (a/b)$, with the conventions $g(a,a)=a$ and $g(a,0)=g(0,a)=0$ for $a>0$, and $g(0,0)=0$.}
\revision{
The transition rates $W^{\omega,\alpha}_{m,n}$ used here are defined using the rank-one projectors associated with an eigenbasis of $\rho$.
In contrast, when the transition rates are defined by Eq.~\eqref{eq:rate_matrix} using projectors onto entire degenerate eigenspaces, replacing the sum by $\sideset{}{^\prime}\sum$ yields a bound in which the mobility $M$ is replaced by $\tilde{M}$ defined in Eq.~\eqref{eq:degenerate_mobility}}
\revision{
\begin{align}
    \norm{\dot{\rho}}_{1} 
    &\leq \frac{2}{\hbar} \sqrt{V_\rho(H)} + \frac{2}{\hbar} \sqrt{V_\rho(H_\calD)}  + \sqrt{2\dot{\sigma}\tilde{M}}.\label{eq:general_VS} \\
    \tilde{M} 
    &:= \sideset{}{^\prime}\sum_{\omega,\alpha,n, m}
    g(W^{\omega,\alpha}_{m,n}p_n, W^{-\omega,\alpha}_{n,m}p_m).
    \label{eq:degenerate_mobility}
\end{align}}%
\revision{The mobility} $M$ is smaller than the activity $A$, so the VS bound is always tighter than the FSS bound.
\revision{
Furthermore, since $\tilde{M}\leq A$ also holds when general spectral projectors are used, the extended bound above is no looser than the FSS bound.}

\section{Main results}\label{results}
\subsection{Speed limit of physical quantities}
Our first main result is the speed limit for the change rate of the physical quantities.
\begin{theorem}\label{th:current_SL}
  The speed limit of the expectation value (current) of a time-independent physical quantity in an open quantum system is given by the following inequality:
  \revision{
  \begin{align}
    \abs{v_{\rho_t}(X)} &\leq \frac{1}{\hbar}  \sqrt{\calF^f_{\rho_t}(H + H_{\calD})} \sqrt{V^f_{\rho_t}(X)} + \sqrt{2\dot{\sigma}\tilde{M}_{\rho_t}(X) } ,\label{eq:current_SL_mobility}\\
    \tilde{M}_{\rho_t}(X) 
    &:= \frac{1}{4} \sideset{}{^\prime}\sum_{\omega,\alpha,a,b} 
    \qty(\Tr[X\hat{\tau}^{\omega,\alpha}_{ab}] - \Tr[X\hat{\tau}^{-\omega,\alpha}_{ba}] )^2 g(W^{\omega,\alpha}_{ab}p_b, W^{-\omega,\alpha}_{ba}p_a) \\
    \hat{\tau}^{\omega,\alpha}_{ab} &:= \frac{\hat{R}^{\omega,\alpha}_{ab}}{\Tr\hat{R}^{\omega,\alpha}_{ab}} , \qquad 
    \hat{R}^{\omega,\alpha}_{ab} := \Pi_a L_{\omega,\alpha}\Pi_b L^\dag_{\omega,\alpha} \Pi_a.
  \end{align}
  }%
  Here, \revision{$v_{\rho_t}(X):=\Tr[X\dot{\rho_t}]$} is the current of the physical quantity $X$, $\calF^f_{\rho_t}(H+H_{\calD})$ is the Fisher information of the state family $\{e^{-i(H+H_{\calD})s}\rho_t e^{i(H+H_{\calD})s}\}_{s\in\mathrm{R}}$ representing the quantum fluctuation of $\rho_t$ with respect to $H+H_{\calD}$. $V^f_{\rho_t}(X)$ is the f-variance of $X$ in $\rho_t$, and $\dot{\sigma}$ is the entropy production rate.
\end{theorem}

\begin{proof}
\revision{
  By definition, the current of a time-independent physical quantity is expressed as follows:
  \begin{align}
    v_{\rho_t}(X) = \Tr[X\dot{\rho}_t].\label{eq:current_def}
  \end{align}
 As shown in Eq.~\eqref{eq:alt_decom}, we can rewrite the GKSL equation as follows:
  \begin{align}
    \dot{\rho}_t = -\frac{i}{\hbar}\qty[H+H_\calD,\rho_t] + \calD_{d}\qty[\rho_t] .\label{eq:Lindblad4}
  \end{align}
  Substituting this into Eq.~\eqref{eq:current_def}, we obtain
  \begin{align}
      \abs{v_{\rho_t}(X)} 
      &= \abs{ \Tr \qty[ -\frac{i}{\hbar}X\qty[H+H_\calD,\rho_t] + X\calD_{d}\qty[\rho_t] ] } \notag\\
      &\leq \abs{ \Tr \qty[ X \cdot\qty(-\frac{i}{\hbar})\qty[H+H_\calD,\rho_t] ]} + \abs{ \Tr\qty[X\calD_{d}\qty[\rho_t]] } \label{eq:decomposition_first_second}.
    \end{align}
    For the first term, the part of $-\frac{i}{\hbar}\qty[H+H_\calD,\rho_t]$ can be regarded as a unitary time evolution with respect to the Hamiltonian $H+H_\calD$. Therefore, using the one-parameter generalized Cramér-Rao inequality \eqref{eq:generalized_CL_ineq}, we can transform it as follows:
    \begin{align}
     \abs{ \Tr \qty[ X \cdot\qty(-\frac{i}{\hbar})\qty[H+H_\calD,\rho_t] ]}
      &\leq \revision{\frac{1}{\hbar}\sqrt{\calF^f_{\rho_t}(H+H_\calD)}\sqrt{V^f_{\rho_t}(X)}.} \label{eq:1para_cramer}
    \end{align}
Next, we evaluate the second term.
First, we define the operator
\begin{align}
    \hat{R}^{\omega,\alpha}_{ab} := \Pi_a L_{\omega,\alpha}\Pi_b L^\dag_{\omega,\alpha} \Pi_a .
\end{align}
Then, \(  W^{\omega,\alpha}_{ab} = \gamma_{\omega,\alpha} \Tr \hat{R}^{\omega,\alpha}_{ab} \) .
Furthermore, when \( \Tr \hat{R}^{\omega,\alpha}_{ab}>0 \), we define
\begin{align}
    \hat{\tau}^{\omega,\alpha}_{ab} := \frac{\hat{R}^{\omega,\alpha}_{ab}}{\Tr\hat{R}^{\omega,\alpha}_{ab}} .
\end{align} 
Note that \( \Tr \hat{R}^{\omega,\alpha}_{ab} = \Tr \hat{R}^{-\omega,\alpha}_{ba} \). 
We then have
\begin{align}
        \Pi_a \calD[\rho_t]\Pi_a 
        &= \sum_{\omega,\alpha} \gamma_{\omega,\alpha} \qty[ \Pi_a L_{\omega,\alpha}\rho L^\dag_{\omega,\alpha} \Pi_a - \frac{1}{2} \Pi_a L^\dag_{\omega,\alpha}L_{\omega,\alpha}\rho \Pi_a - \frac{1}{2} \Pi_a \rho L^\dag_{\omega,\alpha}L_{\omega,\alpha} \Pi_a ] \notag\\
        &= \sum_{\omega,\alpha} \gamma_{\omega,\alpha} \qty[ \sum_{b} p_b \Pi_a L_{\omega,\alpha} \Pi_b L^\dag_{\omega,\alpha} \Pi_a - \sum_b p_a \Pi_a L^\dag_{\omega,\alpha}\Pi_b L_{\omega,\alpha} \Pi_a ] \notag\\
        &= \sum_{\omega,\alpha,b} \gamma_{\omega,\alpha}p_b \hat{R}^{\omega,\alpha}_{ab} - \sum_{\omega,\alpha,b} \gamma_{-\omega,\alpha} p_a \Pi_a L_{\omega,\alpha} \Pi_b L^\dag_{\omega,\alpha}\Pi_a \notag\\
        &= \sum_{\omega,\alpha,b} \gamma_{\omega,\alpha}p_b \hat{R}^{\omega,\alpha}_{ab} - \sum_{\omega,\alpha,b} \gamma_{-\omega,\alpha} p_a \hat{R}^{\omega,\alpha}_{ab} \notag\\
        &= \sum_{\omega,\alpha,b} (\gamma_{\omega,\alpha}p_b - \gamma_{-\omega,\alpha}p_a)\hat{R}^{\omega,\alpha}_{ab}.
\end{align}
Hence, using $\calD_d [\rho_t] = \sum_a \Pi_a \calD[\rho_t]\Pi_a$, we obtain
\begin{align}
    \Tr[X \calD_d[\rho_t]] 
    &= \sum_{\omega,\alpha,a,b} (\gamma_{\omega,\alpha}p_b - \gamma_{-\omega,\alpha}p_a) \Tr[X\hat{R}^{\omega,\alpha}_{ab}] \notag\\
    &= \sum_{\omega,\alpha,a,b} (W^{\omega,\alpha}_{ab}p_b - W^{-\omega,\alpha}_{ba}p_a) \Tr[X\hat{\tau}^{\omega,\alpha}_{ab}]\notag\\
    &= \sideset{}{^\prime}\sum_{\omega,\alpha,a,b}  (W^{\omega,\alpha}_{ab}p_b - W^{-\omega,\alpha}_{ba}p_a) \Tr[X\hat{\tau}^{\omega,\alpha}_{ab}] \\
    &=\frac{1}{2} \sideset{}{^\prime}\sum_{\omega,\alpha,a,b} 
    (W^{\omega,\alpha}_{ab}p_b - W^{-\omega,\alpha}_{ba}p_a) 
    \qty(\Tr[X\hat{\tau}^{\omega,\alpha}_{ab}] - \Tr[X\hat{\tau}^{-\omega,\alpha}_{ba}] ).
\end{align}
}%
\revision{%
Applying the Cauchy--Schwarz inequality, we obtain
\begin{align}
    \abs{\Tr[ X\calD_d[\rho_t]]} 
    &\leq 
    \sqrt{\sideset{}{^\prime}\sum_{\omega,\alpha,a,b}
    \frac{(W^{-\omega,\alpha}_{ba}p_a - W^{\omega,\alpha}_{ab}p_b)^2}{g(W^{\omega,\alpha}_{ab}p_b, W^{-\omega,\alpha}_{ba}p_a)}}\notag \\
    &\times \sqrt{\frac{1}{4} \sideset{}{^\prime}\sum_{\omega,\alpha,a,b} 
    \qty(\Tr[X\hat{\tau}^{\omega,\alpha}_{ab}] - \Tr[X\hat{\tau}^{-\omega,\alpha}_{ba}] )^2 g(W^{\omega,\alpha}_{ab}p_b, W^{-\omega,\alpha}_{ba}p_a)} \\
    &= \sqrt{2 \dot{\sigma} \tilde{M}_{\rho_t}(X)}, \label{eq:classical_bound}
\end{align}
}%
\revision{
where
\begin{align}
    \tilde{M}_{\rho_t}(X) 
    &= \frac{1}{4} \sideset{}{^\prime}\sum_{\omega,\alpha,a,b} 
    \qty(\Tr[X\hat{\tau}^{\omega,\alpha}_{ab}] - \Tr[X\hat{\tau}^{-\omega,\alpha}_{ba}] )^2 g(W^{\omega,\alpha}_{ab}p_b, W^{-\omega,\alpha}_{ba}p_a).
\end{align}
Substituting the bounds in Eqs.~\eqref{eq:1para_cramer} and~\eqref{eq:classical_bound} into Eq.~\eqref{eq:decomposition_first_second},
we obtain Eq.~\eqref{eq:current_SL_mobility}, which completes the proof.
}
\end{proof}

\revision{
\begin{remark}
    Consider the case in which the quantum state $\rho_t$ has a nondegenerate spectrum.
    In this case, each spectral projector can be written as $\Pi_a = \ketbra{a}{a}$.
    Splitting the sum in the definition of the mobility-like quantity $M_{\rho_t}(X)$ into contributions with $a \neq b$ and those with $a=b$ and $\omega \neq 0$, we obtain 
    \begin{align}
        \tilde{M}_{\rho_t}(X) 
        &= \frac{1}{4} \sideset{}{^\prime}\sum_{\omega,\alpha,a,b} 
        \qty(\Tr[X\hat{\tau}^{\omega,\alpha}_{ab}] 
        - \Tr[X\hat{\tau}^{-\omega,\alpha}_{ba}] )^2 
        g(W^{\omega,\alpha}_{ab}p_b, W^{-\omega,\alpha}_{ba}p_a) \\
        &= \frac{1}{4} \sum_{\omega,\alpha,a \neq b} 
        \qty(\Tr[X\hat{\tau}^{\omega,\alpha}_{ab}] 
        - \Tr[X\hat{\tau}^{-\omega,\alpha}_{ba}] )^2 
        g(W^{\omega,\alpha}_{ab}p_b, W^{-\omega,\alpha}_{ba}p_a) \notag \\
        & \quad + \frac{1}{4} \sum_{\omega\neq 0,\alpha,a} 
        \qty(\Tr[X\hat{\tau}^{\omega,\alpha}_{aa}] 
        - \Tr[X\hat{\tau}^{-\omega,\alpha}_{aa}] )^2 
        g(W^{\omega,\alpha}_{aa}p_a, W^{-\omega,\alpha}_{aa}p_a).
    \end{align}
    In the nondegenerate case, the last sum vanishes.
    Furthermore, by appropriately relabeling the summation indices and using the symmetry of $g$, we obtain
    \begin{align}
        \tilde{M}_{\rho_t} (X) 
        &= \frac{1}{4} \sum_{\omega,\alpha,a \neq b} 
        \qty(\Tr[X\hat{\tau}^{\omega,\alpha}_{ab}] 
        - \Tr[X\hat{\tau}^{-\omega,\alpha}_{ba}] )^2 
        g(W^{\omega,\alpha}_{ab}p_b, W^{-\omega,\alpha}_{ba}p_a) \\
        &= \frac{1}{2} \sum_{\omega,\alpha,a \neq b}  \qty((\Tr[X\hat{\tau}^{\omega,\alpha}_{ab}])^2 
        + (\Tr[X\hat{\tau}^{-\omega,\alpha}_{ba}])^2 )
        g(W^{\omega,\alpha}_{ab}p_b, W^{-\omega,\alpha}_{ba}p_a) \notag \\
        & \quad - \frac{1}{4} \sum_{\omega,\alpha,a \neq b} 
        \qty(\Tr[X\hat{\tau}^{\omega,\alpha}_{ab}] 
        + \Tr[X\hat{\tau}^{-\omega,\alpha}_{ba}] )^2 
        g(W^{\omega,\alpha}_{ab}p_b, W^{-\omega,\alpha}_{ba}p_a) \\
        &\leq \sum_{\omega,\alpha,a \neq b}  
        \abs{\Tr[X \hat{\tau}^{\omega,\alpha}_{ab}]}^2
        g(W^{\omega,\alpha}_{ab}p_b, W^{-\omega,\alpha}_{ba}p_a) \\
        &=  \sum_{\omega,\alpha,a \neq b}  
        \abs{\bra{a}X\ket{a}}^2
        g(W^{\omega,\alpha}_{ab}p_b, W^{-\omega,\alpha}_{ba}p_a) 
        =: M_{\rho_t} (X) .
        \label{eq:nondegenerate_mobility}
    \end{align}
    Thus, when the quantum state has a nondegenerate spectrum, the speed limit for the current can be expressed as follows:
    \begin{align}
        \abs{v_{\rho_t}(X)} &\leq \frac{1}{\hbar}  \sqrt{\calF^f_{\rho_t}(H + H_{\calD})} \sqrt{V^f_{\rho_t}(X)} + \sqrt{2\dot{\sigma}M_{\rho_t}(X) } \label{eq:nondegen_SL_current}
    \end{align}
\end{remark}
}
This inequality indicates that the rate of change of the expectation value of any physical quantity $X$ can be bounded by the quantum energy fluctuation and entropy production rate. The entropy term constrains the changes due to the classical dissipative term $\calD_{d}[\rho(t)]$, while the quantum fluctuation term constrains the changes due to the commutator $-\frac{i}{\hbar}[H + H_{\calD}, \rho(t)]$.

\subsection{Improvements of speed limits for quantum states}
Theorem \ref{th:current_SL} provides two improvements of the entropy-based speed limit for quantum states:
\begin{theorem}\label{th:QFI_mobility_bound}
  The speed limit of the quantum state $\rho$ in the open system is given by the following inequality:
  \begin{align}
    \revision{|| \dot{\rho} ||_{1} \leq \frac{1}{\hbar}\sqrt{\calF^f_\rho(H + H_{\cal{D}})} + \sqrt{2\dot{\sigma}\tilde{M}} },\label{eq:QFI_mobility_bound}
  \end{align}
  Furthermore, the following tighter inequality also holds:
  \begin{align}
    \revision{|| \dot{\rho} ||_{1} }
      &\revision{\leq \frac{1}{\hbar} \sqrt{\calF^f_{\rho_t}(H + H_{\cal{D}})} \sqrt{V^f_{\rho_t}(X^\prime)} + \sqrt{2\dot{\sigma} \tilde{M}_{\rho_t}\qty(X^\prime)}}.\label{eq:current_mobility_bound}
  \end{align}
  Here, $X^\prime$ is defined by the diagonal decomposition of $\dot{\rho_t}$ as \revision{$X^\prime := \mathrm{sgn}(\dot{\rho}_t)$},  where $\mathrm{sgn}(Y)$ is defined as $\sum_n \mathrm{sgn}(y_n)\ket{y_n}\bra{y_n}$ for the spectral decomposition of $Y=\sum_n y_n \ketbra{y_n}{y_n}$.
\end{theorem}

\begin{proof}
  Substituting \revision{$X=X^\prime := \mathrm{sgn}\qty(\dot{\rho})$} into Eq.~\eqref{eq:current_SL_mobility}, the left-hand side becomes as follows:
\begin{align}
  \abs{v_{\rho_t}(X^\prime)}= \abs{\Tr[X^\prime\dot{\rho}_t]} = \norm{\dot{\rho}_t}_{1}.
\end{align}
Next, regarding the first term on the right-hand side,
\revision{
\begin{align}
  V^f_{\rho_t}(X^\prime) 
  &= \langle X^\prime_0, X^\prime_0 \rangle_{\rho_t}^f \notag \\
  &=\sum_{ij}p_j f(p_i/p_j)\abs{\bra{i}X^\prime_0\ket{j}}^2 \notag\\
  &\leq \sum_{ij} \frac{p_i+p_j}{2} \abs{\bra{i}X^\prime_0\ket{j}}^2 \notag\\
  &= \sum_i p_i \bra{i}X^{\prime 2}_0\ket{i}\notag\\
  &= \Tr \qty[\rho X^{\prime 2}_0] \notag\\
  &= \Tr\qty[\rho X^{\prime 2}] - \abs{\langle X^\prime \rangle_\rho}^2
  \leq 1
  \label{eq:X_var}
\end{align}}%
\revision{%
holds. 
This is easily understood from the fact that \revision{$X^{\prime 2} \leq I $}.

Next, we bound the second term on the right-hand side, $\tilde{M}_{\rho_t}(X^\prime)$.
Since $X^\prime = \mathrm{sgn}(\dot{\rho})$ satisfies
\begin{align}
    -I \leq X^\prime \leq I, \label{eq:X_prime}
\end{align}
and $\hat{\tau}^{\omega,\alpha}_{ab}\geq 0$ with
$\Tr [\hat{\tau}^{\omega,\alpha}_{ab}]=1$, we have
\begin{align} \qty(\Tr[X^\prime\hat{\tau}^{\omega,\alpha}_{ab}] - \Tr[X^\prime\hat{\tau}^{-\omega,\alpha}_{ba}] )^2 \leq 4.
\end{align}
Therefore,
\begin{align}
    \tilde{M}_{\rho_t}(X^\prime) \leq \sideset{}{^\prime}\sum_{\omega,\alpha,n, m}
    g(W^{\omega,\alpha}_{m,n}p_n, W^{-\omega,\alpha}_{n,m}p_m) = \tilde{M}.
\end{align}}%
\revision{Combining the above estimates, we obtain
\begin{align}
  || \dot{\rho} ||_{1} 
  &\leq \frac{1}{\hbar} \sqrt{\calF^f_{\rho_t}(H + H_{\calD})} \sqrt{V^f_{\rho_t}(X^\prime)} + \sqrt{2\dot{\sigma}\tilde{M}_{\rho_t}\qty(X^\prime)} \label{eq:current_bound}\\
  &\leq \frac{1}{\hbar} \sqrt{\calF^f_{\rho_t}(H + H_{\calD})} + \sqrt{2\dot{\sigma}\tilde{M}} .
\end{align}%
}
\end{proof}
\revision{%
\begin{remark}
    In particular, consider the case in which the quantum state
    $\rho_t$ has a nondegenerate spectrum.
    Equation~\eqref{eq:X_prime} implies that
    $\abs{\bra{a}X^\prime\ket{a}}\leq 1$.
    Therefore, Eq.~\eqref{eq:nondegenerate_mobility} yields
\begin{align}
    M_{\rho_t}(X^\prime) \leq \sum_{\omega,\alpha,n \neq m}
    g(W^{\omega,\alpha}_{m,n}p_n, W^{-\omega,\alpha}_{n,m}p_m) = M,
\end{align}
 where $M$ is defined in Eq.~\eqref{eq:mobility}.
 Thus, when the quantum state has a nondegenerate spectrum, the speed limit for quantum state can be expressed as follows:
 \begin{align}
     || \dot{\rho} ||_{1} 
  &\leq \frac{1}{\hbar} \sqrt{\calF^f_{\rho_t}(H + H_{\calD})} \sqrt{V^f_{\rho_t}(X^\prime)} + \sqrt{2\dot{\sigma}M_{\rho_t}\qty(X^\prime)} \label{eq:nondegen_current_bound}\\
  &\leq \frac{1}{\hbar} \sqrt{\calF^f_{\rho_t}(H + H_{\calD})} + \sqrt{2\dot{\sigma}M} \label{eq:nondegen_QFI_bound}.
 \end{align}
\end{remark}
}

The inequality Eq.~\eqref{eq:QFI_mobility_bound} bounds the rate of state change\revision{ $|| \dot{\rho} ||_{1}$} by two quantities. One is the quantum fluctuation \revision{$\calF^f_{\rho}(H+H_{\calD})$ }that bounds the changes due to the commutator $-\frac{i}{\hbar}[H + H_{\calD}, \rho(t)]$, and the other is the entropy production that bounds the changes due to the diagonal dissipative term \revision{$\calD_{d}[\rho(t)]$}.
Since the diagonal dissipative term corresponds to the ``classical" dissipation, our result corresponds to decomposing the infinitesimal time evolution of an open quantum system into a unitary direction and a classical direction, bounding the unitary direction as if it were the time evolution of an isolated system, and bounding the classical direction as if it were the time evolution of a classical system. 
Our result implies that, in such a decomposition, the speed in the unitary direction is determined by the quantum Fisher information, which represents the quantum fluctuation, and the speed in the classical direction is determined by the entropy production rate, which represents the strength of dissipation. 

This understanding is fully consistent with the classical speed limits \cite{SFS,Ohzeki} that bound the speed of classical system evolution by the entropy production and the speed limits for isolated systems given in the context of asymmetry \cite{coherence_QSL_isolated} that bound the speed of time evolution of an isolated system is determined by the quantum Fisher information of the Hamiltonian driving the isolated system, although the isolated QSL is described with the Bures distance.
Therefore, intuitively, the speed limit for open quantum systems is a kind of sum of the speed limit for unitary time evolution and the speed limit for classical systems. 
However, in the case of the open quantum systems, the unitary time evolution differs from that of a simple isolated system by the amount of the effective Hamiltonian $H_{\calD}$ induced by the heat bath.
We will see this evaluation is very tight in the next subsection.

\subsection{Comparison with the previous entropy-based speed limits}

\revision{
For the SLD inner product, Eq.~\eqref{eq:QFI_mobility_bound} takes the form
\begin{align}
    || \dot{\rho} ||_{1} \leq \frac{1}{\hbar}\sqrt{\calF_\rho(H + H_{\cal{D}})} + \sqrt{2\dot{\sigma}\tilde{M}}. \label{eq:SLDQFI_bound}
\end{align}
This bound can be viewed as the result of replacing the two energy-variance terms $2\sqrt{V_\rho(H)}+2\sqrt{V_\rho(H_\calD)}$ in Eq.~\eqref{eq:general_VS} with the quantum Fisher information term $\sqrt{\calF_\rho(H+H_{\calD})}$.
This substitution eliminates the effect of classical energy fluctuations from the evaluation of the inequality~\eqref{eq:VS}. In fact, the following relation holds:}%
\revision{
\begin{align}
    \sqrt{\calF_\rho(H+H_{\calD})} \le 2\sqrt{V_\rho(H+H_{\calD})} \le 2\sqrt{V_\rho(H)} + 2\sqrt{V_\rho(H_{\calD})}.
\end{align}
}%
Therefore, Eq.~\eqref{eq:SLDQFI_bound} is always tighter than \revision{Eq.~\eqref{eq:general_VS}} \revision{(and consequently Eq.~\eqref{eq:FSS})}.

To further visualize such comparisons, we perform a numerical analysis for a two-level system. 
Under the settings described in Appendix \ref{Ap:numerical_setting}, we compute the time evolution of $\rho_t$, and calculate the variance, Fisher information, activity, mobility, and other quantities at each time step to compare and evaluate the various bounds and \revision{$\norm{\dot{\rho}}_{1}$}.
\revision{The behavior of the bounds under different parameter values is also presented in Appendix~\ref{Ap:numerical_setting}.
}

\revision{
We show the improvements over the FSS bound~\eqref{eq:FSS} achieved by the VS bound~\eqref{eq:VS} and our bound (QFI bound~\eqref{eq:nondegen_QFI_bound} and current bound~\eqref{eq:nondegen_current_bound}) in Fig.~\ref{fig:kaizen}, with our bound evaluated using the SLD inner product.
In this example, the improvement from replacing the activity with the mobility is not so significant, and a large gap remains between the VS bound and the actual speed $\norm{\dot{\rho}}_1$.
By contrast, as can be seen from the figure, our bound provides a significant improvement.
Indeed, taking the FSS bound as a reference, the improvement achieved by our bound clearly exceeds that achieved by the VS bound.
This difference arises from replacing the term containing the energy variance in the existing bounds with the quantum fluctuation (quantum Fisher information) of the energy. 
Whereas the quantum Fisher information reflects the property of quantum fluctuations, the variance also includes a contribution from thermal fluctuations, corresponding to the quantity $C_\rho(H)$ introduced in Eq.~\eqref{variance-decom}.
Our new bound excludes this contribution and is therefore tighter than both the FSS and VS bounds.
}
\revision{
In addition, as shown in Fig.~\ref{fig:H_driven}, our bound remains tighter than the variance-based bounds even when the Hamiltonian is time-dependent.
}%
\revision{%
Altogether, }our new bound\revision{ improves on the previous bound} in both the fluctuation term and the entropy-production term.

Furthermore, substituting \revision{$\mathrm{sgn}(\dot{\rho})$} directly into the current bound for $X$ in the inequality \eqref{eq:current_SL_mobility} provides an even better upper bound on the rate of state change (=Eq.~\eqref{eq:current_mobility_bound}). 
In this case, as well, it is important to note that the first term of Eq.~\eqref{eq:current_mobility_bound} is proportional to the Fisher information, representing the effect of quantum fluctuations in energy. 

Finally, Fig.~\ref{fig:ryoushi} compares the classical term \revision{$\sqrt{2\dot{\sigma}M}$ with $\|\dot{\rho}\|_{1}$}, excluding the Fisher information term. 
As shown in the figure, the classical term alone does not serve as an upper bound for \revision{ $\norm{\dot{\rho}}_{\mathrm{1}}$ (=smaller than $\norm{\dot{\rho}}_{\mathrm{1}}$). }
Note that this quantity works as an upper bound on the rate of change of probability distributions in the case of the classical master equation. 
Therefore, from our results, it can be clearly seen that quantum fluctuations (Fisher information) positively influence the rate of change of quantum states.

\revision{
\subsection{Behavior of the speed limit in the presence of degeneracy of quantum states}
In situations where the quantum state exhibits degeneracy, the quantum speed limit can become temporarily loose.
This issue originates from the component $H_\calD$ of the external influence, which can be regarded as an effective Hamiltonian.
The effective Hamiltonian $H_\calD$ is defined as in Eq.~\eqref{eq:H_D}.
This definition excludes contributions from pairs of states sharing the same eigenvalue. 
However, in realistic dynamical scenarios involving interactions between a quantum system and a thermal bath, one must also consider cases in which eigenvalues approach each other arbitrarily closely, either before or after becoming exactly degenerate.
For a two-level system, numerical calculations in such regimes exhibit behavior as shown in Fig.~\ref{fig:shukutai}.}

\revision{
This issue arises both for the Fisher-information-based bound introduced in this work and for conventional variance-based bounds.
Nevertheless, even in these regimes, the Fisher-information-based bound remains tighter than the variance-based one.
As illustrated in Fig.~\ref{fig:shukutai}, the variance-based speed limit develops a pronounced peak and appears to diverge in practice, whereas the Fisher-information-based speed limit does not exhibit such strong divergence.
}

\revision{
The origin of this difference can be clarified by examining the respective definitions of variance and Fisher information.
The loosening of the bound is caused by the denominator term $(p_i - p_j)^2$ appearing in the effective Hamiltonian $H_\calD$, which becomes extremely small when eigenvalues approach each other.
The variance defined with respect to a general inner product is given by Eq.~\eqref{eq:variance_f}, while the Fisher information is given by Eq.~\eqref{eq:fisher-state-family}.
Calculating the variance, we obtain
\eq{
V^f_\rho(H_\calD) &= \expect{H_{\calD 0},H_{\calD 0}}^f_{\rho_t} \\
&= \expect{H_{\calD},H_{\calD}}^f_{\rho_t}\\
&= \Tr \qty[ H_\calD \mathbb{J}^f_{\rho_t}(H_\calD) ]\\
&= \sum_{ij,i\neq j} \frac{\hbar^2\abs{\bra{j}\calD[\rho]\ket{i}}^2}{(p_i-p_j)^2}p_j f\qty(\frac{p_i}{p_j}).
}
From this expression, it can be seen that a term $(p_i-p_j)$ appears in the denominator, and thus the bound generally diverges when the difference between two eigenvalues becomes extremely small (however, when the two eigenvalues coincide exactly, those terms are excluded by definition, and the bound remains finite). 
On the other hand, calculating the Fisher information, we obtain
\eq{
\calF^f_{\rho_t}(H+H_\calD)
&=\sum_{ij}\frac{(p_i-p_j)^2}{p_j f\qty(\frac{p_i}{p_j})}\abs{\bra{i}(H+H_\calD)\ket{j}}^2 .
}
Here, since the terms with $i=j$ vanish, we consider only the terms with $i\neq j$
\eq{
\abs{\bra{i}H+H_\calD\ket{j}}^2  &=
\abs{\bra{i}H\ket{j}}^2 + \abs{\bra{i}H_\calD\ket{j}}^2 + \bra{i}H\ket{j}\bra{j}H_\calD\ket{i} + \bra{i}H_\calD\ket{j}\bra{j}H\ket{i}.
}
In particular, for the term $\abs{\bra{i}H_\calD\ket{j}}^2$,
\eq{
\abs{\bra{i}H_\calD\ket{j}}^2 = \frac{\hbar^2\abs{\bra{i}\calD[\rho]\ket{j}}^2}{(p_i-p_j)^2}.
}
In this case, the factor $(p_i-p_j)^2$ in the denominator, which causes the divergence, is canceled by the corresponding factor $(p_i-p_j)^2$ in the numerator, and hence the divergence is suppressed.

\revision{
Furthermore, the following proposition holds.

\begin{proposition}
The mobility $M_{\rho_t}(X^\prime)$ in the current bound for the two-level system considered in Fig.~\ref{fig:shukutai} vanishes at the time $t_*$ of closest approach to the maximally mixed state.
In other words, the classical contribution from the heat bath to the rate of state change tends to zero as $t$ approaches $t_*$.
\end{proposition}
\begin{proof}
    The quantum state $\rho$ and its time derivative $\dot{\rho}$ can be expressed in the Bloch representation as
    \begin{align}
        \rho=\frac{1}{2}(I+\bm{r}\cdot \bm{\sigma}), \quad \dot{\rho} = \frac{1}{2} \dot{\bm{r}}\cdot \bm{\sigma} ,
    \end{align}
    where $\bm{r}$ is the Bloch vector and $\bm{\sigma}$ is the vector of Pauli matrices. 
    Since we are interested in the evolution as the state passes near the maximally mixed state, we assume $0<\abs{\bm{r}}<1$ in what follows.
    We also assume $\abs{\dot{\bm{r}}}>0$ and introduce the normalized vectors and their inner product as
    \begin{align}
        \bm{u} := \frac{\bm{r}}{\abs{\bm{r}}}, \quad 
        \bm{n} := \frac{\dot{\bm{r}}}{\abs{\dot{\bm{r}}}}, \quad 
        c = \bm{u}\cdot \bm{n}.
    \end{align}
    With these definitions, $X^\prime$ and the spectral projectors $\Pi_{\pm}$ of $\rho$ can be written as
    \begin{align}
        X^\prime=\mathrm{sgn}(\dot{\rho}) = \bm{n}\cdot \bm{\sigma}, \quad \Pi_{\pm}=\frac{1}{2}(I\pm \bm{u}\cdot\bm{\sigma}) .
    \end{align}
    Using the above results, we rederive Eq.~\eqref{eq:nondegenerate_mobility} for the mobility appearing in the current bound.
    First, this mobility can be expressed as follows:
    \begin{align}
        M_{\rho_t} (X^\prime)  &= \sum_{\omega,\alpha,a\ne b} \abs{\bra{a}X^\prime\ket{a}}^2 g(W^{\omega,\alpha}_{a,b}p_b, W^{-\omega,\alpha}_{b,a}p_a).
    \end{align}
    Since the two eigenstates $\ket{+}$ and $\ket{-}$ of $\rho$ satisfy
    \begin{align}
        \bra{+}X^\prime\ket{+}=c,\quad \bra{-}X^\prime\ket{-}=-c,
    \end{align}
    the mobility $M_{\rho_t}\qty(X^\prime)$ can be written as
    \begin{align}
        M_{\rho_t}\qty(X^\prime)=c^2 M.
    \end{align}
    Near the maximally mixed state, where $\abs{\bm{r}}$ is close to zero, the mobility $M_{\rho_t}(X^\prime)$ is suppressed by the factor $c^2$, determined by the inner product between $\bm{r}/\abs{\bm{r}}$ and $\dot{\bm{r}}/\abs{\dot{\bm{r}}}$, and vanishes when these two vectors are orthogonal.
    We now consider the near-degenerate case of interest, namely a trajectory that passes near, but does not reach, the maximally mixed state.
    Let $t_*$ denote the time of closest approach to the maximally mixed state, and assume $\abs{\bm{r}(t_*)}>0,\abs{\dot{\bm{r}}(t_*)}>0$.
    Since $\abs{\bm{r}}^2$ has a local minimum at this time, we have
    \begin{align}
        0=\left. \frac{d}{dt}\abs{\bm{r}(t)}^2 \right|_{t_*} = 2\bm{r}(t_*)\cdot \dot{\bm{r}}(t_*).
    \end{align}
    Hence, $c=0$ at $t=t_*$, implying that $M_{\rho_{t_*}}(X^\prime)=0$.
\end{proof}
}
}

\begin{figure}
  \centering
  \includegraphics[width=12cm]{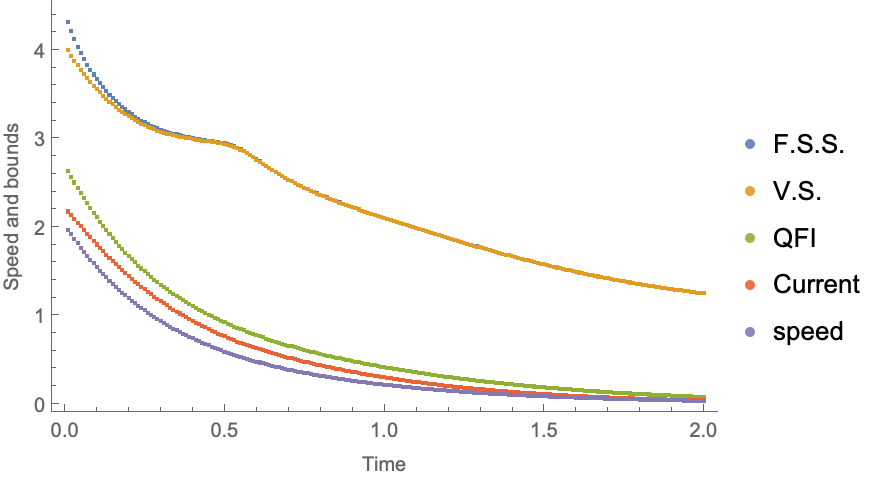}
   \caption{
   \revision{This figure compares the previous bounds (the FSS~\eqref{eq:FSS} and VS~\eqref{eq:VS} bounds), the bounds obtained in this work (the QFI~\eqref{eq:nondegen_QFI_bound} and current~\eqref{eq:nondegen_current_bound}  bounds with SLD inner product), and the speed of quantum-state evolution $\|\dot{\rho}_t\|_1$.
   As described in Appendix~\ref{Ap:numerical_setting} for this and all subsequent figures, we compute the time evolution of the speed and each bound for a set of randomly chosen initial states and average the resulting values over these states at each time.
   Although the VS bound is tighter than the FSS bound, a large gap remains between the VS bound and the actual speed $\|\dot{\rho}_t\|_1$.
   By contrast, the QFI and current bounds are substantially tighter, on average, than these previous variance-based bounds.}}
   \label{fig:kaizen}
\end{figure}

\begin{figure}
    \centering
    \includegraphics[width=12cm]{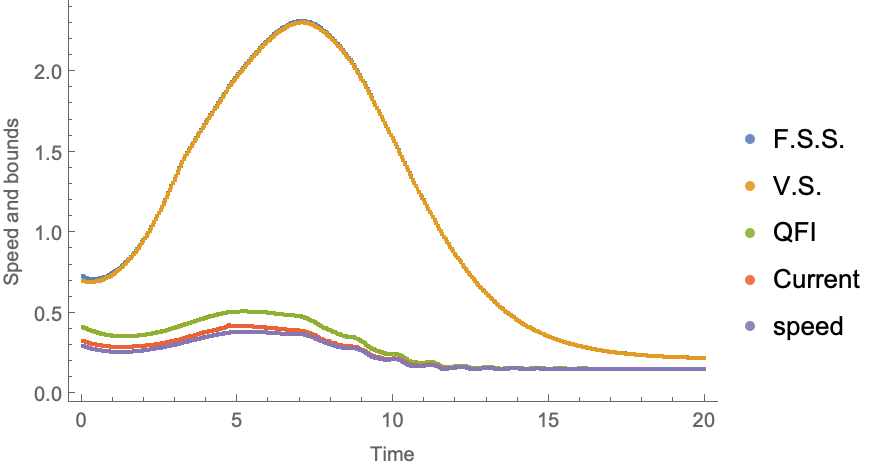}
    \caption{\revision{Behavior of the quantum speed limits in the case of a driven Hamiltonian. 
    This figure compares the previous bounds (the FSS~\eqref{eq:FSS} and VS~\eqref{eq:VS} bounds), the bounds obtained in this work (the QFI~\eqref{eq:nondegen_QFI_bound} and current~\eqref{eq:nondegen_current_bound}  bounds with SLD inner product), and the speed of quantum-state evolution $\|\dot{\rho}_t\|_1$.
    Even in this case, the bounds evaluated by the QFI are tighter than those evaluated by the variance.}}
    \label{fig:H_driven}
\end{figure}

\begin{figure}
  \centering
  \includegraphics[width=12cm]{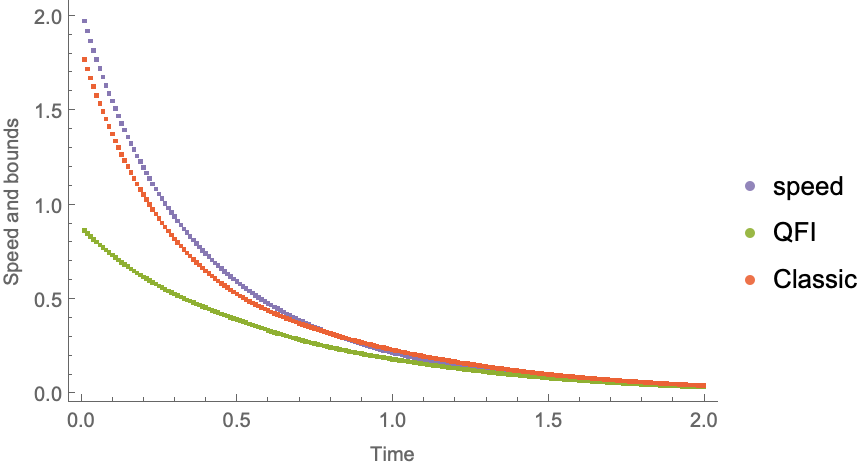}
   \caption{Comparison of \revision{$\norm{\dot{\rho}_t}_{1}$} with the quantum term \revision{(=the QFI term: the first term of Eq.~\eqref{eq:nondegen_QFI_bound} with SLD inner product)} and the classical term \revision{(=the entropy production term: the second term of Eq.~\eqref{eq:nondegen_QFI_bound})} in the QFI bound. 
   Both of these terms are smaller than \revision{$\norm{\dot{\rho}_t}_{1}$ in most of the time region in the graph}, and thus each of these two terms does not work as an upper bound on \revision{$\norm{\dot{\rho}_t}_{1}$} individually. Only when we combine them can we provide an upper bound on \revision{$\norm{\dot{\rho}_t}_{1}$}. This fact indicates that both quantum and classical effects are crucial in estimating the rate of state change in open quantum systems.}
   \label{fig:ryoushi}
\end{figure}

\begin{figure}
  \centering
  \includegraphics[width=12cm]{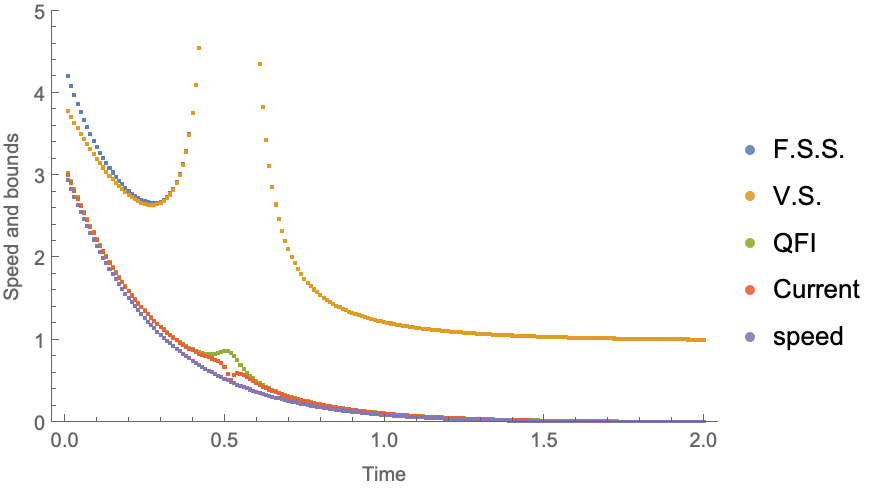}
  \caption{\revision{
  Behavior of the quantum speed limits in a two-level system when the quantum state approaches degeneracy.
  The two variance-based speed limits (FSS and VS) exhibit strong divergences near the degeneracy points, whereas the Fisher-information-based bounds do not exhibit such strong divergences.
  Furthermore, the current bound becomes even tighter as the quantum state passes through the vicinity of a degeneracy, because the mobility $M_{\rho}(X^\prime)$ entering the classical-fluctuation term approaches zero.
  }}
  \label{fig:shukutai}
\end{figure}

\begin{figure}
    \centering
    \includegraphics[width=12cm]{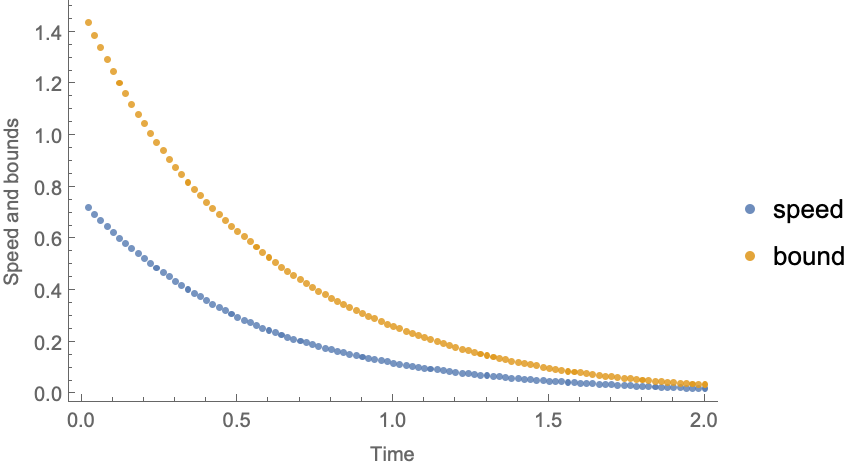}
    \caption{\revision{Numerical calculation of the speed limit for the current Eq.~\eqref{eq:nondegen_SL_current} associated with the Pauli $X$ operator.}}
    \label{fig:X_current}
\end{figure}
\begin{figure}
    \centering
    \includegraphics[width=12cm]{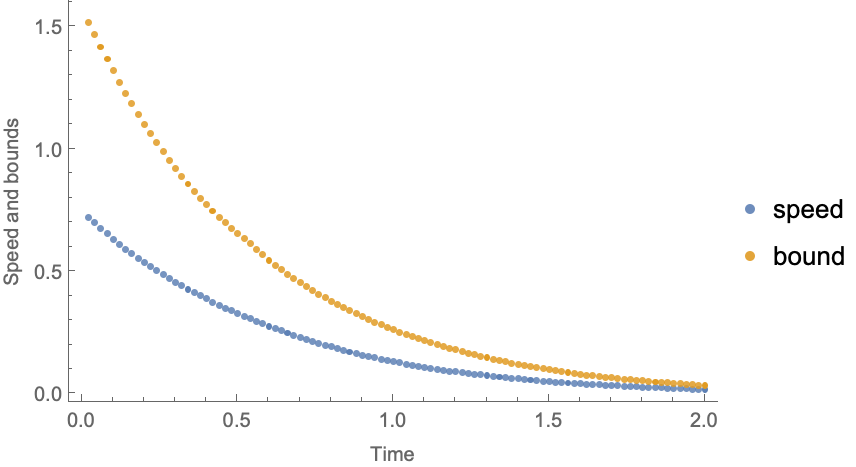}
    \caption{\revision{Numerical calculation of the speed limit for the current Eq.~\eqref{eq:nondegen_SL_current} associated with the Pauli $Y$ operator.}}
    \label{fig:Y_current}
\end{figure}
\begin{figure}
    \centering
    \includegraphics[width=12cm]{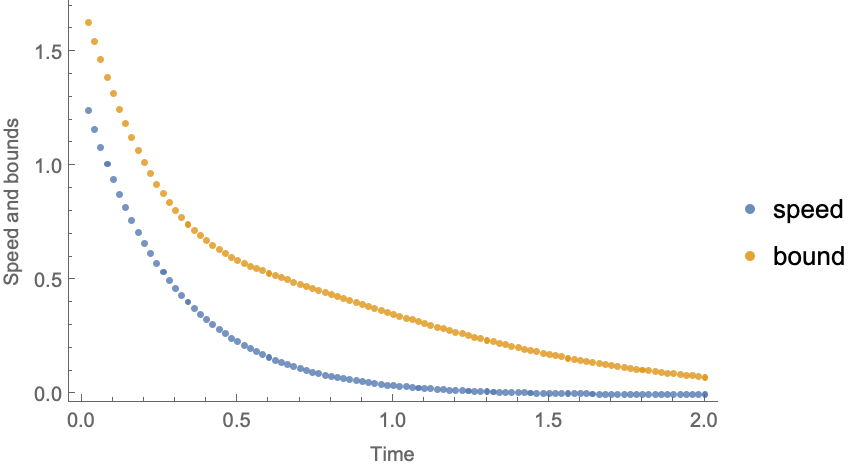}
    \caption{\revision{Numerical calculation of the speed limit for the current Eq.~\eqref{eq:nondegen_SL_current} associated with the Pauli $Z$ operator.}}
    \label{fig:Z_current}
\end{figure}

\section{Conclusion}
In this paper, we derived the speed limit of the expectation value of a time-independent physical quantity in an open quantum system.
\revision{
\begin{align}
  \abs{v_{\rho_t}(X)} \leq\frac{1}{\hbar} \sqrt{\calF^f_{\rho_t}(H + H_{\calD})} \sqrt{V^f_{\rho_t}(X)} + \sqrt{2\dot{\sigma}\tilde{M}_{\rho_t}\qty(X)} .\label{eq:current_SL2}
\end{align}
}%
The proof of this theorem uses the generalized Cramér-Rao inequality for one-parameter quantum state families and the evaluated expectation value by the diagonal term of the dissipation term of the physical quantity using the Cauchy-Schwarz inequality. 
Furthermore, we showed that there exists a bound that is tighter than the FSS bound Eq.~\eqref{eq:FSS} and the VS bound Eq.~\eqref{eq:general_VS}.
\begin{align}
  \revision{|| \dot{\rho} ||_{1} \leq \frac{1}{\hbar}\sqrt{\calF^f_\rho(H + H_{\calD})} + \sqrt{2\dot{\sigma}\tilde{M}}}.
\end{align}
These results are described using Fisher information, showing that the speed limit of the current of a physical quantity and the quantum state speed limit are limited by quantum fluctuations.
\revision{
Moreover, with an appropriate choice of inner product, such as the SLD inner product, the above inequality yields a bound that is tighter than both the FSS bound~\eqref{eq:FSS} and the VS bound~\eqref{eq:general_VS}.
}

Regarding the relationships between coherence and the speed limit,
there is another inequality about the trade-off between the current, the entropy production rate, and coherence between degeneracies \cite{TF,Funo-Tajima}. We leave the relationship between this inequality and the present result as future work.

\acknowledgments
HT acknowledges ``MEXT KAKENHI Grant-in-Aid for Transformative
Research A'' as B ``Quantum Energy Innovation” Grant Numbers 24H00830 and 24H00831, and JST PRESTO No. JPMJPR2014, JST MOONSHOT No. JPMJMS2061. K. F. acknowledges support from JSPS KAKENHI Grant Number JP23K13036 and JST ERATO Grant Number JPMJER2302, Japan.

In this letter, the authors obtained the help of ChatGPT 6Pro and DeepL Translator to correct English expressions naturally.
These helped only in the final stages of revising the text in the manuscript; the authors alone derived the research ideas and results and wrote the original manuscript.

\bibliographystyle{quantum}
\bibliography{paper1}

\appendix

\section{Definition of transition rates using projection operators and consistency with the conventional formulation}
\label{Ap:projection_op}
The transition rates $W^{\omega,\alpha}_{mn}$ are defined by
\begin{align}
    W^{\omega,\alpha}_{mn}:= \gamma_{\omega,\alpha} \Tr (\Pi_m L_{\omega,\alpha}\Pi_n L_{\omega,\alpha}^\dag ) .
\end{align} 
A straightforward calculation using this definition shows that
$W^{-\omega,\alpha}_{nm}
= \gamma_{-\omega,\alpha}
\Tr(\Pi_m L_{\omega,\alpha}\Pi_n L_{\omega,\alpha}^\dag)$:
\eq{
     W^{-\omega,\alpha}_{nm} 
     &= \gamma_{-\omega,\alpha} \Tr (\Pi_n L_{-\omega,\alpha}\Pi_m L_{-\omega,\alpha}^\dag )\notag \\
     &= \gamma_{-\omega,\alpha} \Tr (\Pi_n L^\dag_{\omega,\alpha}\Pi_m L_{\omega,\alpha})\notag \\
     &= \gamma_{-\omega,\alpha} \Tr (\Pi_m L_{\omega,\alpha}\Pi_n L_{\omega,\alpha}^\dag ).
}
We then obtain
\begin{align}
    -\beta \dot{Q} 
    &= \sum_{\omega,\alpha} \gamma_{\omega,\alpha}\Tr [L_{\omega,\alpha}\rho L^\dag_{\omega,\alpha}]\beta \omega \notag \\
    &= \sum_{\omega,\alpha} \gamma_{\omega,\alpha} \Tr [\sum_m \Pi_m L_{\omega,\alpha}\sum_n p_n \Pi_n L^\dag_{\omega,\alpha}]\beta \omega \notag  \\
    &= \sum_{\omega,\alpha,n,m} p_n \gamma_{\omega,\alpha} \Tr [\Pi_m L_{\omega,\alpha}\Pi_n L^\dag_{\omega,\alpha}]\beta \omega \notag \\
    &= \sum_{\omega,\alpha,n,m} W^{\omega,\alpha}_{mn} p_n \beta \omega.
\end{align}
Here, the first equality follows from $\comm{L_{\omega,\alpha}}{H}=\omega L_{\omega,\alpha}$.
Since the terms satisfying both $(m=n)\land (\omega=0)$ vanish, the sum can equivalently be replaced by the restricted sum$\sum^\prime$, which excludes these terms.
Similarly,
\begin{align}
    \dot{S} 
    &= -\sum_{\omega,\alpha}\gamma_{\omega,\alpha}\Tr[L_{\omega,\alpha}\rho L^\dag_{\omega,\alpha} \ln \rho] + \sum_{\omega,\alpha}\gamma_{\omega,\alpha}\Tr[L^\dag_{\omega,\alpha}L_{\omega,\alpha}\rho \ln \rho] \notag \\
    &= -\sum_{\omega,\alpha,n,m}\gamma_{\omega,\alpha}p_n \ln p_m \Tr[L_{\omega,\alpha} \Pi_n L^\dag_{\omega,\alpha} \Pi_m] + \sum_{\omega,\alpha,n,m}\gamma_{\omega,\alpha}p_n \ln p_n \Tr[L^\dag_{\omega,\alpha} \Pi_m L_{\omega,\alpha}\Pi_n] \notag \\
    &= -\sum_{\omega,\alpha,n,m} W^{\omega,\alpha}_{mn} p_n \ln p_m + \sum_{\omega,\alpha,n,m} W^{\omega,\alpha}_{mn} p_n \ln p_n.
\end{align}
In particular, when the quantum state has a nondegenerate spectrum, each spectral projector has rank one, i.e., $\Pi_n=\ketbra{n}{n}$.
The transition rates then reduce to $W^{\omega,\alpha}_{mn} =\gamma_{\omega,\alpha}\abs{\bra{m}L_{\omega,\alpha}\ket{n}}^2$, in agreement with the definition in Ref.~\cite{FSS}.

\section{Definition of Fisher information }\label{Ap:Fisher_general}
Using the $f$-inner product defined above, we can express the definition of Fisher information for many inner products at once.

First, we introduce a \revision{Hermitian} operator $L$ called a logarithmic derivative as follows.
\begin{definition}
  For a quantum state family $\qty{\rho_t}_t$ that is differentiable with respect to $t$, the logarithmic derivative $L$ is defined as follows:
  \begin{align}
    \frac{\partial \rho_t}{\partial t}= \mathbb{J}^f_{\rho_t}\qty(L). \label{eq:log_derivative_def}
  \end{align}
\end{definition}
By definition, $L$ depends on $\qty{\rho_t}_t$ and $f$, but it is customary to omit $f$ and write only $L$. Furthermore, it is possible to construct $L$ that satisfies this definition.

\begin{theorem}
  The logarithmic derivative $L$ that satisfies $\mathbb{J}_{\rho_t}^f(L) = \frac{\partial \rho_t}{\partial t}$ can be constructed as follows when $\rho=\sum_i p_i \ketbra{i}{i}$ is diagonalized:
  \begin{align}
    L = \sum_{i,j} \frac{1}{p_j f(\frac{p_i}{p_j})}\bra{i} \frac{\partial \rho_t}{\partial t} \ket{j} \ketbra{i}{j}.  \label{eq:logarithmic_derivative}
  \end{align}
\end{theorem}
\begin{proof}
    By expanding $\mathbb{J}_{\rho_t}^f(L) $ in a Maclaurin series and using $\rho=\sum_i p_i \ketbra{i}{i}$,
    \begin{align*}
      \mathbb{J}_{\rho_t}^f(L) 
      &= \sum_m^\infty \frac{f^{(m)}(0)}{m!} \rho^m L \rho^{1-m} \\
      &= \sum_m^\infty \frac{f^{(m)}(0)}{m!} \sum_{ij} p_j \qty(\frac{p_i}{p_j})^m \bra{i} L \ket{j} \ketbra{i}{j}\\
      &= \sum_{ij} p_j \sum_m^\infty \frac{f^{(m)}(0)}{m!}  \qty(\frac{p_i}{p_j})^m \bra{i} L \ket{j} \ketbra{i}{j}\\
      &= \sum_{ij} p_j f\qty(\frac{p_i}{p_j}) \bra{i} L \ket{j} \ketbra{i}{j}.
\end{align*}
    Also, if $\mathbb{J}_{\rho_t}^f(L) = \frac{\partial \rho_t}{\partial t}$, then
    $\bra{i}\mathbb{J}_{\rho_t}^f(L) \ket{j} = \bra{i} \frac{\partial \rho_t}{\partial t} \ket{j}$, so
    \begin{align*}
      \bra{i}\mathbb{J}_{\rho_t}^f(L) \ket{j} = p_j f\qty(\frac{p_i}{p_j}) \bra{i} L \ket{j} = \bra{i} \frac{\partial \rho_t}{\partial t} \ket{j} .
    \end{align*}
    Thus,
    \begin{align*}
      \bra{i} L \ket{j} = \frac{1}{p_j f\qty(\frac{p_i}{p_j})} \bra{i} \frac{\partial \rho_t}{\partial t} \ket{j} .
    \end{align*}
    Therefore,
    \begin{align}
      L = \sum_{i,j} \frac{1}{p_j f(\frac{p_i}{p_j})}\bra{i} \frac{\partial \rho_t}{\partial t} \ket{j} \ketbra{i}{j} \label{eq:logarithmic_derivative2}
    \end{align}
    is obtained.
\end{proof}
The Fisher information is defined using the logarithmic derivative.
\begin{definition}
  For a quantum state family $\qty{\rho_t}_t$ that is differentiable with respect to $t$, the Fisher information $J^f_{\rho_t}$ is defined as follows:
  \begin{align}
    J^f_{\rho_t} = \langle L, L \rangle_{\rho_t}^f. \label{eq:Fisher_def1}
  \end{align}
  Here, defining $\qty(\mathbb{J}^f_{\rho_t})^{-1}$ as the inverse \revision{map} of $\mathbb{J}^f_{\rho_t}$, Eq.~ \eqref{eq:Fisher_def1} becomes:
  \begin{align}
    J^f_{\rho_t} = \Tr\qty[ \frac{\partial \rho_t}{\partial t} \qty(\mathbb{J}^f_{\rho_t})^{-1} \qty( \frac{\partial \rho_t}{\partial t} )  ]. \label{eq:Fisher_def2}
  \end{align}
\end{definition}
  This property means that the Fisher information represents the square of the absolute value of the speed of the state with respect to the parameter $t$ in the quantum state family $\qty{\rho_t}_t$. Therefore, when using the Fisher information in our main results, we take the square root to align the dimension of the speed.

\section{Resource Theory of Asymmetry}\label{Ap:RTA}
Resource theory is a theory that defines a state that can be easily prepared (free state) and an operation that can be easily performed (free operation), and defines a state that cannot be realized by combining these as a resource, and quantitatively considers the resource.
Resource theory is not uniquely determined, and different resource theories are obtained depending on how free states and free operations are defined.
To quantitatively evaluate the amount of resources, a resource measure is defined as a function as shown below \cite{RevModPhys.91.025001}.
\begin{definition}
  The resource measure is defined as a function $R(\rho)$ that satisfies the following conditions:
  \begin{enumerate}
    \item $R(\rho) \geq 0$.
    \item \revision{$\rho$ is a free state $\Rightarrow R(\rho)$ = 0}.
    \item $\calE$ is a free operation $\Rightarrow R(\rho) \geq R(\calE(\rho))$.
  \end{enumerate}
  Furthermore, when the restriction \revision{``$R(\rho) = 0 \Rightarrow \rho$ is a free state" }is imposed on condition 2, $R(\rho)$ is called a faithful resource measure.
\end{definition}
It should be noted that the quantity defined as a resource measure depends on the resource to be considered.

Next, we describe the resource theory of asymmetry (RTA) \cite{Marvian_thesis,Marvian_distillation,Marvian2022operational,Takagi_skew,skew_resource, YT,YT2,Kudo,ST} used to interpret our main results.
RTA is a resource theory that deals with dynamics with symmetry and treats quantum fluctuations of conserved quantities as resources. Because it involves the fundamental concept of symmetry in physics, it has a wide range of applications, from quantum measurements \cite{WAY_RToA1,WAY_RToA2,TN, TTK, ET2023, nakajima2024speed}, unitary gate implementation \cite{TSS,TSS2, TS, TTK, ET2023, nakajima2024speed}, coherence distribution \cite{Marvian2018,Lostaglio2018}, quantum speed limits \cite{coherence_QSL_isolated}, quantum error correction codes \cite{TS,e-EKZhou,e-EKYang,Liu1,Liu2, TTK}, quantum thermodynamics \cite{TTK,tajima2024gibbs}, to the black hole information escape problem \cite{TS, TTK}.
In this section, we introduce the foundation of RTA and commonly used resource indicators. Since there is no standard text on RTA, we formulate and introduce various results based on \cite{Marvian_thesis} and mix the contents of \cite{Marvian_thesis,Marvian_distillation,Marvian2022operational,Takagi_skew,skew_resource, YT,YT2,Kudo,ST}.

RTA is a resource theory that considers asymmetry as a resource. First, we consider a group to discuss symmetry and asymmetry. The definition of a group is given as follows.
\begin{definition}
  A set $G$ is a group if the operation ($\cdot$) on $G$ satisfies the following conditions:
  \begin{itemize}
    \item Associativity: $\forall a,b,c \in G, (a\cdot b)\cdot c = a\cdot (b\cdot c)$.
    \item Identity element: $\exists e \in G, \forall a \in G, a\cdot e = e\cdot a = a$.
    \item Inverse element: $\forall a \in G, \exists a^{-1} \in G, a\cdot a^{-1} = a^{-1}\cdot a = e$.
  \end{itemize}
\end{definition}
A group is an abstract concept that generalizes reversible operations. The existence of the identity element represents an operation that does nothing, and the existence of the inverse element represents an operation that cancels out the original state for any operation. Furthermore, groups play an important role in dealing with symmetry. For example, the threefold symmetric group and the rotation and symmetrical movements of an equilateral triangle can be associated with each other, but the shape of an equilateral triangle does not change due to the action of a group (operation such as rotation). Properties that are invariant under the action of a group are called symmetries.
Next, we associate a group that deals with such symmetry with a quantum system. The transformation of an isolated quantum system is represented by a unitary operator, as can be seen from the Schrödinger equation. We associate this unitary operator with a group as follows.

\begin{definition}
  A map \revision{$U: G \rightarrow \calU(\calH)$} from a group $G$ to unitary operators on a Hilbert space $\calH$ is called a unitary representation of the group $G$ when it satisfies the following property:
  \begin{align}
    U(g_1)\cdot U(g_2) = U(g_1\cdot g_2) .\label{eq:group_rep}
  \end{align}
\end{definition}
From Eq. (\ref{eq:group_rep}), we find that the unit element and the inverse element of the group $G$ are related as follows:
\begin{align}
  U(e) = I,\quad U(g^{-1}) = U(g)^{-1}.
\end{align}
Furthermore, we define a projective unitary representation as an extension of the unitary representation. First, from Eq. (\ref{eq:group_rep}), for any $g_1,g_2 \in G$,
\begin{align}
  U(g_1)U(g_2) \rho U(g_2)^{-1}U(g_1)^{-1} = U(g_1g_2) \rho U(g_1g_2)^{-1}.
\end{align}
It is understood that this holds. We call the extension of $U(g)$ that satisfies this a projective unitary representation.
\begin{definition}
  For a map \revision{$U: G \rightarrow \calU(\calH)$} from a group $G$ to a set of unitary operators on a Hilbert space, if there exists a complex number $e^{i\omega(g_1,g_2)}$ that satisfies
  \begin{align}
    U(g_1)U(g_2) = e^{i\omega(g_1,g_2)}U(g_1g_2).
  \end{align}
  If a complex number $e^{i\omega(g_1,g_2)}$ exists that satisfies this, $U(\cdot)$ is called a projective unitary representation of the group $G$. Here, $\omega(g_1,g_2)$ is a real number defined for two elements $g_1,g_2$ of $G$, and the set of complex numbers $\qty{e^{i\omega(g_1,g_2)}}$ is called the \revision{factor system} of the projective unitary representation.
\end{definition}

Using the above definitions, we define free states and free operations in RTA.
These are respectively called symmetric states and covariant operations, and are defined as follows.

\begin{definition}
  A quantum state $\rho$ is called a symmetric state when it satisfies
  \begin{align}
    U(g) \rho U(g)^{-1} = \rho
  \end{align}
  for any element $g$ of the group $G$.
\end{definition}
\begin{definition}
  A CPTP-map $\calE$ is called a covariant operation when it satisfies
  \begin{align}
    \calE(U_{in}(g)\rho U_{in}(g)^{-1})  = U_{out}(g)\calE(\rho)U_{out}(g)^{-1}
  \end{align}
  for any element $g$ of the group $G$. Here, $U_{in}(g),U_{out}(g)$ are the projective unitary representations of $g$ on the input and output Hilbert spaces of the map $\calE$, respectively.
\end{definition}
In this paper, we consider the case where the group $G$ is the real number $\mathbb{R}$ or the one-dimensional unitary group $U(1)$.

A state that is not a symmetric state, such as the one described above, is called a resource state or an asymmetric state. In this case, it is necessary to measure how much resources a resource state has using a resource measure. For the group $\mathbb{R}$ or the one-dimensional unitary group $U(1)$, the SLD Fisher information is often used as a good resource measure.

\begin{definition}
  Using a periodic Hermitian operator $H$, the projective unitary representation of the group $U(1)$ can be written as $U_t=e^{-iHt}$. In this case, the Fisher information for the family of quantum states $\qty{U_t \rho U_t^\dag}$ generated by the unitary operator $U_t$ is defined as
  \begin{align}
    \calF_\rho^f(H) = \left. J^f_{U_t \rho U_t^\dag} \right|_{t=0}.
  \end{align}
  
\end{definition}
The following theorem holds for the $\calF_\rho^f(H)$ defined above.

\begin{theorem}
  If the spectral decomposition of \revision{$\rho$ is $\sum_i p_i \ketbra{i}{i}$}, then $\calF_\rho^f(H)$ can be calculated as follows:
  \revision{
  \begin{align}
    \calF_\rho^f(H) = \sum_{i,j} \frac{(p_i - p_j)^2}{p_j f(\frac{p_i}{p_j})} \abs{ \bra{i} H \ket{j} }^2 . \label{eq:fisher_f}
  \end{align}
  }
\end{theorem}
\revision{
Among the Fisher informations defined in this way, the case corresponding to $f=(1+x)/2$ is called the SLD Fisher information, and is denoted by $\calF_\rho(H)$ in the following.
}
\revision{
Furthermore, the following theorem holds for the Fisher information \cite{min_V_Yu}:
\SLDdecomposition*
Together with Theorem \ref{theo:SLD_Fisher2}, this shows that the SLD Fisher information serves as a good indicator of quantum fluctuations.
And although the quantum Fisher information with any other monotone metric satisfies Luo's criteria, the minimal decomposition as in Theorem \ref{theo:SLD_Fisher2} does not hold.
In this sense, the SLD Fisher information is particularly suitable as an indicator of the quantum fluctuation.
}%
Resource indicators for the case of a general group $G$ are also being clarified. In particular, for Lie groups, it is known that the Fisher information matrix becomes a resource indicator \cite{Kudo}. In the case of discrete finite groups, it has recently been shown that the absolute value of the \revision{characteristic function's log} determines the transformation rate of the i.i.d. state, which is a resource indicator \cite{ST}.
\revision{This study suggests that the larger the resource of asymmetry is, the faster the speed of time evolution can be made.}

\section{Cramér-Rao inequality}\label{Ap:Cramer-Rao}
The one-parameter generalized Cramér-Rao inequality is an inequality concerning the time change of the expectation value of a physical quantity, and is expressed as follows.
\begin{theorem}
  For a smooth quantum state family $\qty{\rho_t}_t$ and a physical quantity $A$, the following inequality holds for the time change of the expectation value of $A$:
  \begin{align}
    \abs{\frac{\partial}{\partial t} \langle A \rangle_{\rho_t} } \leq \sqrt{J_{\rho_t}^f}\sqrt{V_{\rho_t}^f(A)}.
    \label{eq:generalized_CL_ineq}
  \end{align}
\end{theorem}
\begin{proof}
Define $A_0$ as
  \begin{equation*}
    A_0 \coloneqq A - \langle A \rangle_{\rho_t},
  \end{equation*}
then
  \begin{equation*}
    \langle A \rangle_{\rho_t} \Tr \qty[ \frac{\partial \rho_t}{\partial t} ] = \langle A \rangle_{\rho_t} \frac{\partial}{\partial t} \Tr \qty[ \rho_t ] = 0,
  \end{equation*}
  so
  \begin{align*}
    \abs{\frac{\partial}{\partial t} \langle A \rangle_{\rho_t} } 
    &= \abs{ \Tr \qty[ A \frac{\partial \rho_t}{\partial t} ] } = \abs{ \Tr \qty[ A_0 \frac{\partial \rho_t}{\partial t} ] }.
  \end{align*}
  Here, using the \revision{Hermitian} operator $L$ as the logarithmic derivative,
  \begin{align}
    \abs{ \Tr \qty[ A_0 \frac{\partial \rho_t}{\partial t} ] } = \abs{ \Tr \qty[ A_0 \mathbb{J}_{\rho_t}^f(L) ] }
  \end{align}
  can be done. Since this is an $f$-inner product of $A_0$ and \revision{$L$}, by the Cauchy-Schwarz inequality
  \begin{align*}
    \abs{\langle A_0, L \rangle_{\rho_t}^f }
    &\leq \sqrt{ \langle A_0, A_0 \rangle_{\rho_t}^f } \sqrt{ \langle L, L \rangle_{\rho_t}^f } \\
    &= \sqrt{V_{\rho_t}^f(A)} \sqrt{J_{\rho_t}^f},
  \end{align*}
  the proof is completed.
\end{proof}

\section{Setting of numerical calculation system}\label{Ap:numerical_setting}

\revision{To compare the bounds numerically, we simulated the following GKSL equation using Mathematica:}%
\begin{widetext}
\begin{align}
  \dot{\rho}_t =  -\frac{i}{\hbar}[H,\rho_t] + 
  \gamma_-\qty[ \sigma_- \rho_t \sigma_+ - \frac{1}{2}\qty\big{ \sigma_+ \sigma_- , \rho_t } ] + \gamma_+ \qty[ \sigma_+ \rho_t \sigma_- - \frac{1}{2}\qty\big{ \sigma_- \sigma_+ , \rho_t } ]. \label{eq:Lindblad5}
\end{align}
\end{widetext}%
\revision{%
where $\sigma_+$ and $\sigma_-$ are the jump operators corresponding to absorption and emission, respectively, and $\gamma_- =\gamma_0 (N(\omega_0)+1)$ and $\gamma_+ =\gamma_0 N(\omega_0)$ are the corresponding rates, with  $N(\omega_0)=1/(e^{\beta \omega_0}-1)$.
We set $\hbar = 1$ throughout.}

\revision{In each simulation, we sampled 100 initial states uniformly at random from the Bloch ball, computed the time evolution for each, and averaged the resulting trajectories of the speed and the bounds over the samples.
For the two-level dynamics considered here, the set of initial states whose trajectories reach the maximally mixed state $\rho_t=I/2$ within the finite simulation interval has zero volume in the Bloch ball.
Consequently, for initial states sampled uniformly from the Bloch ball, the density operator has a nondegenerate spectrum throughout this interval with probability one.
We therefore used the nondegenerate expressions for the mobility, given in Eqs.~\eqref{eq:nondegenerate_mobility} and~\eqref{eq:mobility}, in our numerical calculations.}

\revision{For the case of a time-independent Hamiltonian (Figs.\ref{fig:kaizen} and \ref{fig:ryoushi}-\ref{fig:Z_current}), we used the following setting:}
\begin{itemize}
  \item Hamiltonian $H = \frac{\omega_0}{2} \sigma_z =\mqty[ 0.5 & 0 \\ 0 & -0.5 ]$ ($\omega_0 = 1$).
  \item Inverse temperature $\beta \omega_0 = -\log \frac{\gamma_+}{\gamma_-} = 0.288$ (\revision{we set $N(\omega_0)=3$}).
  \item \revision{Coupling strength} $\gamma_0 = 0.5$.
\end{itemize}
\revision{For the case of a driven Hamiltonian (Fig.~\ref{fig:H_driven}), the settings are as follows:

\eq{
 H(t) = \frac{\epsilon(t)}{2}\mqty(\cos \theta(t) & \sin \theta(t) \\ \sin \theta(t) & -\cos \theta(t)) 
}
\eq{
\epsilon(t) = \epsilon_0 + (\epsilon_1 - \epsilon_0)\sin^2 \qty(\frac{\pi t}{2T}),
\qquad
\theta(t) = \pi \qty(\frac{t}{T}-1)
}
with $\epsilon_0 = 0.4, \epsilon_1 = 10$, $T= 20$, and $\beta =1$. 
The GKSL equation retains the form of Eq.~\eqref{eq:Lindblad5}, with $\sigma_\pm$ and $\gamma_\pm$ replaced by
\eq{
\tau_\pm(t) = \frac{1}{2}\qty[\cos\theta(t)\,\sigma_x
              - \sin\theta(t)\,\sigma_z \pm i \sigma_y],
}
\eq{
\gamma_-(t) = \tilde{\gamma}\,\epsilon(t)\qty[N(\epsilon(t))+1],
\qquad
\gamma_+(t) = \tilde{\gamma}\,\epsilon(t)\,N(\epsilon(t)),
}
where $N(\epsilon)=1/(e^{\beta \epsilon}-1)$ and $\tilde{\gamma} = 0.2$.
The operators $\tau_+(t)$ and $\tau_-(t)$ satisfy $[H(t),\tau_\pm(t)] = \pm\epsilon(t)\,\tau_\pm(t)$, and thus describe absorption and emission between the energy eigenstates of $H(t)$, respectively. 
Since
$\gamma_+(t)/\gamma_-(t) = e^{-\beta\epsilon(t)}$, the detailed balance
condition~\eqref{eq:detail_balance} is satisfied at each time.

Finally, we investigated how the bounds behave when the mean excitation
number $N(\omega_0)$ and the coupling strength $\gamma_0$ are varied in
the dynamics under the time-independent Hamiltonian
(Figs.~\ref{fig:g01_n3}-\ref{fig:g05_n7}). For all parameter values examined, our
bound is substantially tighter than the FSS and VS bounds.
}
\newpage

\begin{figure}
  \begin{minipage}[c]{0.5\hsize}
    \centering
    \includegraphics[scale=0.5]{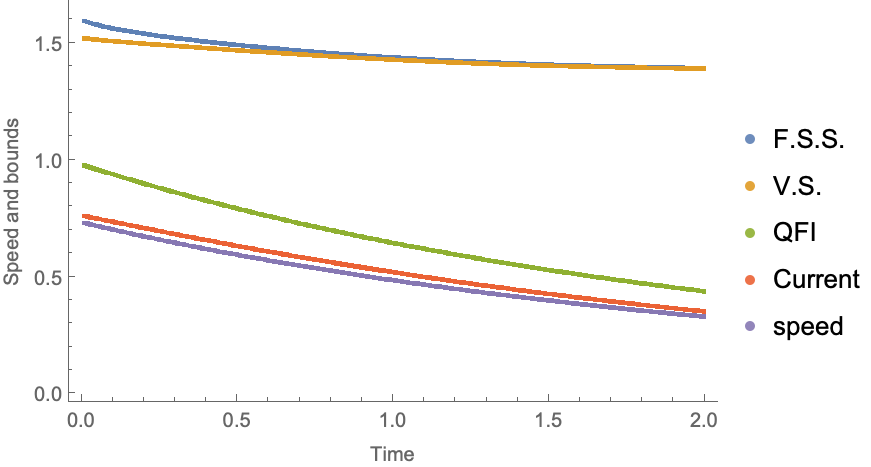}
    \caption{\revision{Numerical calculations for $\gamma_0=0.1,N=3$}.}
    \label{fig:g01_n3}
  \end{minipage}
  \begin{minipage}[c]{0.5\hsize}
    \centering
    \includegraphics[scale=0.5]{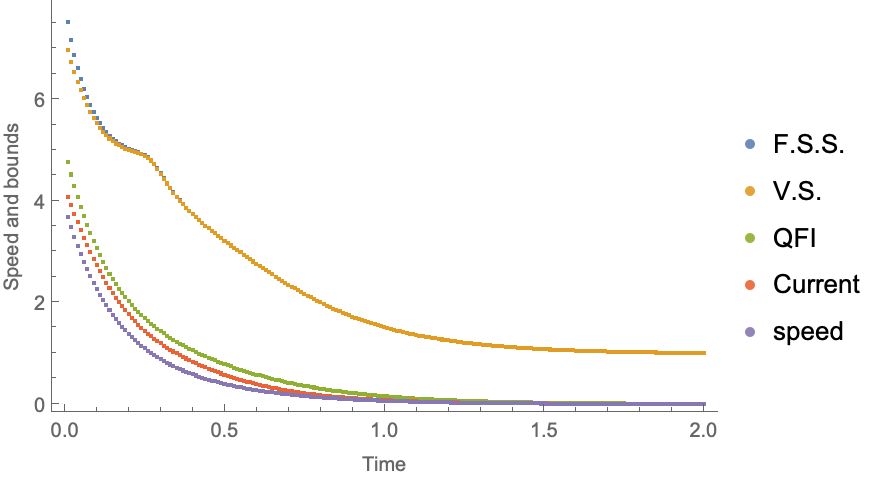}
    \caption{\revision{Numerical calculations for $\gamma_0=1,N=3$}.}
    \label{fig:g1_n3}
  \end{minipage}
\end{figure}

\begin{figure}
  \begin{minipage}[c]{0.5\hsize}
    \centering
    \includegraphics[scale=0.5]{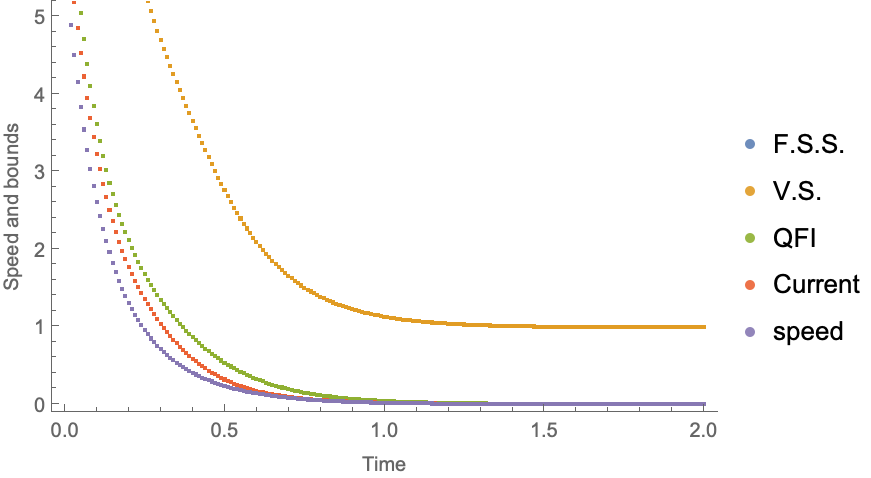}
    \caption{\revision{Numerical calculations for $\gamma_0=
    1.5,N=3$.}}
    \label{fig:g1p5_n3}
  \end{minipage}
  \begin{minipage}[c]{0.5\hsize}
    \centering
    \includegraphics[scale=0.5]{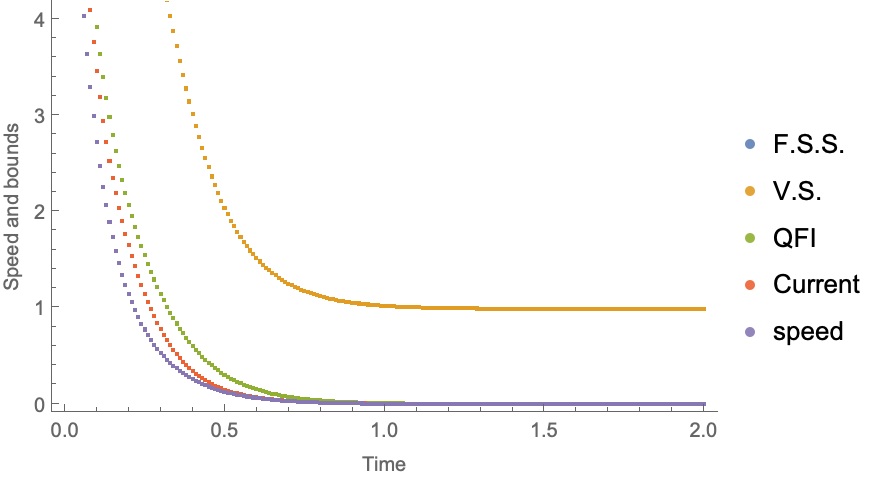}
    \caption{\revision{Numerical calculations for $\gamma_0=2,N=3$.}}
    \label{fig:g2p5_n3}
  \end{minipage}
\end{figure}

\begin{figure}
  \begin{minipage}[c]{0.5\hsize}
    \centering
    \includegraphics[scale=0.5]{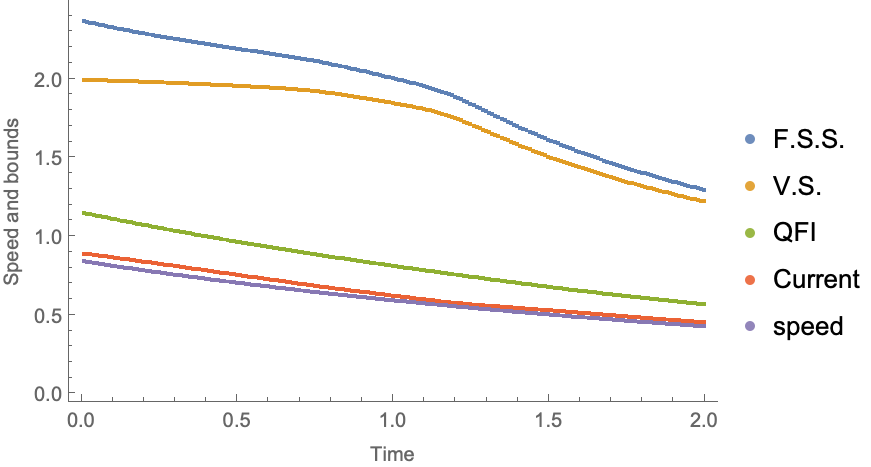}
    \caption{\revision{Numerical calculations for $\gamma_0=
    0.5,N=0.01$.}}
    \label{fig:g05_n001}
  \end{minipage}
  \begin{minipage}[c]{0.5\hsize}
    \centering
    \includegraphics[scale=0.5]{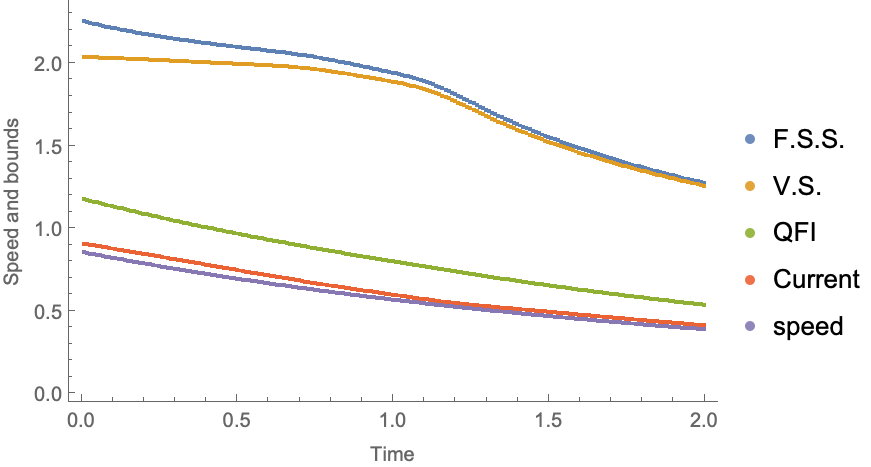}
    \caption{\revision{Numerical calculations for $\gamma_0=0.5,N=0.1$.}}
    \label{fig:g05_n01}
  \end{minipage}
\end{figure}

\begin{figure}
  \begin{minipage}[c]{0.5\hsize}
    \centering
    \includegraphics[scale=0.5]{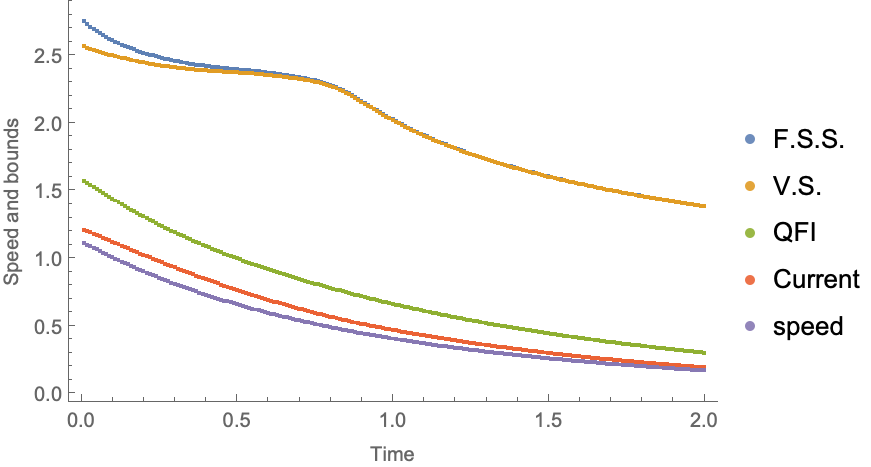}
    \caption{\revision{Numerical calculations for $\gamma_0=
    0.5,N=1$.}}
    \label{fig:g05_n1}
  \end{minipage}
  \begin{minipage}[c]{0.5\hsize}
    \centering
    \includegraphics[scale=0.5]{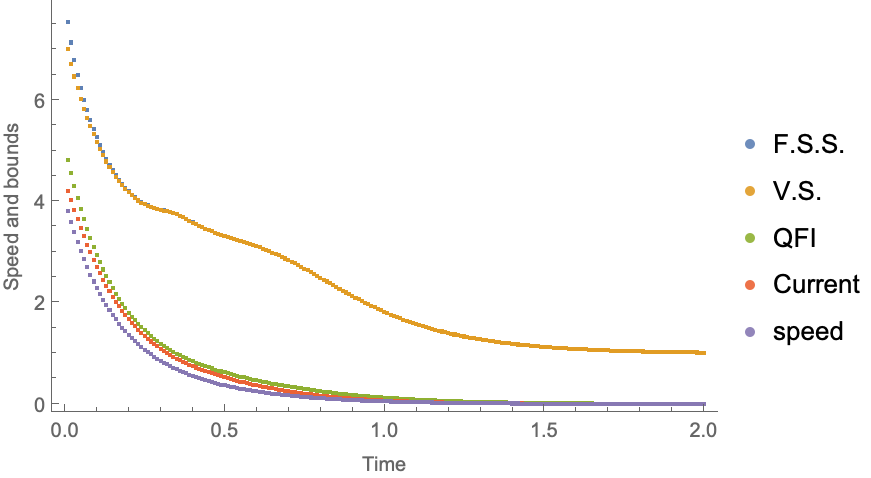}
    \caption{\revision{Numerical calculations for $\gamma_0=0.5,N=7$.}}
    \label{fig:g05_n7}
  \end{minipage}
\end{figure}

\end{document}